\documentclass[twoside,11pt]{article}

\usepackage{jmlr2e}
\usepackage{url}
\usepackage{amsmath}
\usepackage{mathtools}
\usepackage{subcaption}
\usepackage{todonotes}
\usepackage{url}
\usepackage{algorithm}
\usepackage{algorithmicx}
\usepackage{algpseudocode}
\usepackage{booktabs}
\usepackage{pifont}
\usepackage{units}
\usepackage{placeins}

\newcommand{\cmark}{\ding{51}}%
\newcommand{\xmark}{\ding{55}}%

\ShortHeadings{Real-Time Community Detection in Large Social Networks on a Laptop}{Chamberlain, Levy-Kramer, Humby, Deisenroth}
\firstpageno{1}

\begin{document}

\title{Real-Time Community Detection in Large Social Networks on a Laptop}

\author{\name Ben Chamberlain 
\email b.chamberlain14@imperial.ac.uk \\
       \addr Department of Computing\\
       Imperial College London\\
       London SW7 2AZ, UK
       \AND
       \name Josh Levy-Kramer 
       \email josh@starcount.com \\
       \addr Starcount Insights \\
       2 Riding House Street \\
       London W1W 7FA
       \AND
       \name Clive Humby
       \email clive@starcount.com \\
       \addr Starcount Insights \\
       2 Riding House Street \\
       London W1W 7FA
	   \AND
	   \name 
       Marc Peter Deisenroth 
       \email m.deisenroth@imperial.ac.uk \\
       \addr Department of Computing\\
       Imperial College London\\
       London SW7 2AZ, UK
       }

%\editor{xxxxx}

\maketitle
\begin{abstract}%   <- trailing '%' for backward compatibility of .sty file
%\todo[inline]{is the method for generating ground truth from the Twitter network worth mentioning in the abstract?}
%\todo[inline]{The positioning of the figures and the references still need to be sorted}

% what is the problem
For a broad range of research, governmental and commercial applications it is important to understand the allegiances, communities and structure of key players in society. One promising direction towards extracting this information is to exploit the rich relational data in digital social networks (the social graph). 
% what others do
As social media data sets are very large, most approaches make use of distributed computing systems for this purpose. Distributing graph processing requires solving many difficult engineering problems, which has lead some researchers to look at single-machine solutions that are faster and easier to maintain.
% What we do
In this article, we present a single-machine real-time system for large-scale graph processing that allows analysts to interactively explore graph structures.
%key idea
The key idea is that the aggregate actions of large numbers of users can be compressed into a data structure that encapsulates user similarities while being robust to noise and queryable in real-time.
% how do we solve it? mention results and set-up later. give key idea here
We achieve single-machine real-time performance by compressing the neighbourhood of each vertex using minhash signatures and facilitate rapid queries through Locality Sensitive Hashing. These techniques reduce query times from hours using industrial desktop machines operating on the full graph to milliseconds on standard laptops. 
% why it is important
Our  method allows exploration of strongly associated regions (i.e. communities) of large graphs in real-time on a laptop. 
It has been deployed in software that is actively used by social network analysts and offers another channel for media owners to monetise their data, helping them to continue to provide free services that are valued by billions of people globally.

\end{abstract}

%\begin{keywords}
%\end{keywords}

% \todo[inline]{Redraw Fig. 9 and 11. Make sure circles and boxes and arrows make sense  and consistently mean the same.}
% \todo[inline]{Address all orange boxes. Then delete them.}
% \todo[inline]{Fix the references (consistency)}

\section{Introduction}
\label{sec:intro}

% \todo[inline]{Can we show an illustration of social network snippet?}
% \todo[inline]{Can we illustrate the high-level approach? (Feed in social network, return community). Show the main steps in a diagram}

% community detection introduction
Algorithms to discover groups of associated entities from relational (networked) data sets are often called \emph{community detection} methods. They come in two forms: global methods, which partition the entire graph and local methods that look for vertices that are related to an input vertex and only work on a small part of the graph. We are concerned with community detection on large graphs that runs on a single commodity computer. To achieve this we combine the two approaches, using local community detection to identify an interesting region of a graph and then applying global community detection to help understand that region.

%The techniques that we employ are well known. 
% % the problem with marketing on social networks
% \todo[inline]{Should I kill this paragraph?}
% In recent years, Digital Social Networks (DSN) have experienced unprecedented growth. Twitter, Facebook, Instagram and YouTube are household names on six continents and are a daily part of billions of people's lives. Facebook, the world's most pervasive DSN, commands a valuation exceeding 300 billion dollars, while smaller companies are often valued in the tens of billions. However, valuations and the ability to provide free services are being challenged by changing attitude towards advertising. Over-saturation and new uses of digital media have depleted the effectiveness of long-established models like web banners and sponsored content. 

% The solution to monetising social networks
% Our work demonstrates an alternative monetisation channel for DSNs; 
Our focus is on community detection using social media data. Social media data provides a record of global human interactions at a scale that is hitherto unprecedented. Discovering communities in the social graph has a large number of governmental and industrial applications, which include: security, where analysts explore a network looking for groups of potential adversaries; social sciences, where queries can establish the important relationships between individuals of interest; e-commerce, where queries reveal related products or users; marketing, where companies seek to optimise advertising channels or celebrity endorsement portfolios. These applications do not disrupt user experience in the way that sponsored links or feed advertising do offering an alternative means for social media providers to continue to offer free services. 

% A sketch of how it might work
As an illustration of a commercial application of community detection using Twitter data, take a company that wants to trade in a new geographic region. To do this they need to understand the region's competitors, customers and marketing channels. Using our system they input the Twitter handles for their existing products, key people, brands and endorsers, and in real-time receive the accounts closely related to their company in that market. The output is automatically structured into groups (communities) such as media titles, sports people and other related companies. Analysts examine the results and explore different regions by changing the input accounts. We show a high level illustration for the drinks brand Diageo in Figure~\ref{fig:overview}.

\begin{figure}[tb]
\includegraphics[width=\textwidth]{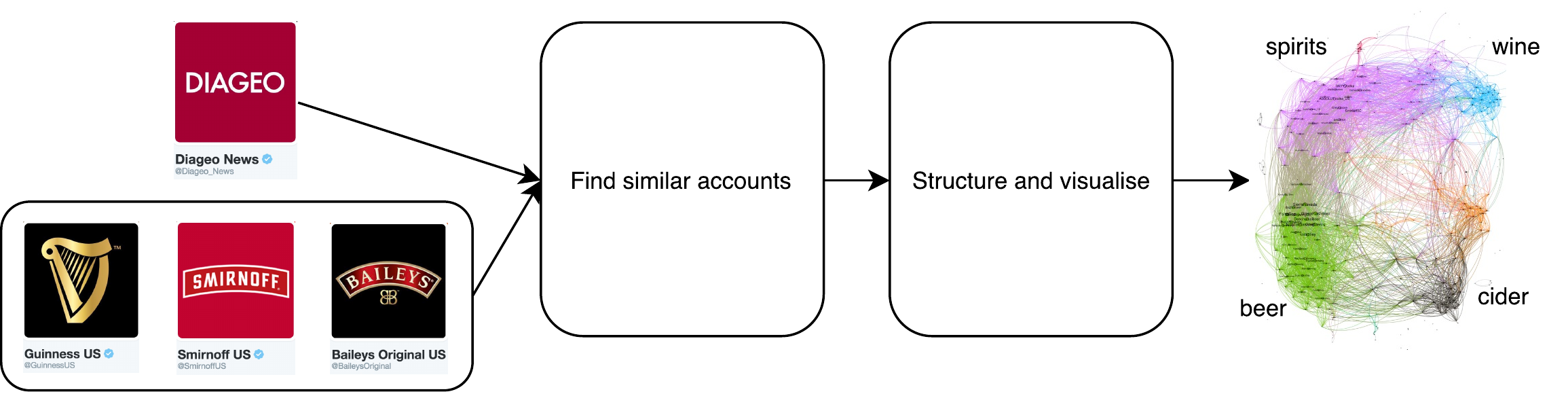}
\caption{An example of how our system can be used: Diageo want to explore the market (competitors, customers, associations etc.) around their brand. They feed in information about themselves (``seeds''). In this example the seeds are the company itself (Diageo) and some of their major brands (Smirnoff, Baileys and Guinness). Our systems finds accounts that are similar/related to the seeds and then structures the similar accounts into communities.}
\label{fig:overview}
\end{figure}

% Introduction to the social graph
Throughout this paper we refer to graphs. In this context a graph is a collection of vertices and edges connecting them. A graph is usually written $G(V,E)$ and a graph with weighted edges as $G(V,E,W)$. A network is a richer structure than a graph, comprising a graph and a collection of metadata describing the vertices and/or edges of the graph. A community is a collection of vertices $C \subset V$ that share many more edges than would be expected from a random subset of vertices. In the context of the Twitter graph, a vertex $V$ is a Twitter account and an (undirected) edge $E$ between $V_i, V_j$ exists if $V_i$ Follows $V_j$ or $V_j$ Follows $V_i$. A community might be the set of Twitter accounts belonging to machine learning researchers. In addition to the Twitter graph, the Twitter \emph{network} also includes metadata associated with the accounts (e.g., name, description) and edges (eg. creation time, direction). 

% The neighbourhood graph
Our algorithm focusses exclusively on the properties of the graph.
We are particularly interested in the neighbourhood graph. The neighbourhood graph of a vertex consists of the set of all vertices that are directly connected to it, irrespective of the edge direction. DSN neighbourhood graphs can be very large; In Twitter, the largest have almost 100 million members (as of June 2016). We propose that robust associations between social network accounts can be reached by considering the similarity of their neighbourhood graphs. 
% homophily 
This proposition relies on the existence of homophily in social networks. The homophily principle states that people with similar attributes are more likely to form relationships \citep{McPherson2001}. Accordingly, social media accounts with similar neighbourhood graphs are likely to have similar attributes.

% Our constraints
We seek to build a system that: (1) produces high quality communities from very noisy data. (2) Is robust to failure and does not require engineering support. (3) is parsimonious with the time of its users. The first constraint leads us to use the neighbourhood graph as the unit of comparison between vertices. The neighbourhood graph is generated by the actions of large numbers of independent users in contrast to features like text content or group memberships, which are usually controlled by a single user.  The second requirements leads us to search for a single-machine solution and the third prescribes a real-time system (or as close as possible). 
% High performance needed
High performance is vital as analysts wish to interact with the data, combining the results of previous experiments to inform new ones. The difference between real-time and `quite quick' is important. Real-time response is primary amongst the reasons that interactive program languages like Python and R have replaced compiled languages like C++ as the tools of choice for data analysts. We aim to offer similar improvements in usability.

% current methods
Currently, no tool exists that provides real-time analysis of large graphs on a single commodity machine. Existing methods to analyse local community structure in large graphs either rely on distributed computing facilities or incur excessive run-times making them impractical for exploratory and interactive work \citep{Clauset2005, Bahmani2011}. 
% our tool
In this article, we describe our real-time analysis tool for detecting communities in large graphs using only a laptop. We focus on a 700 million user Twitter network. However, our work is more generally applicable as it does not rely upon Twitter-specific data, only the graph structure, and we provide some results from Facebook to demonstrate this.

% how it works
There are two core problems to solve: (1) The graph must be fit into the memory of a single (commodity) machine. (2) Many neighbourhood graphs containing up to 100 million vertices must be compared in milliseconds. The first step to solving these problems is to compress the neighbourhood graphs into fixed-length minhash signatures. Minhash signatures vastly reduce the size of the graph while at the same time encoding an efficient estimation of the Jaccard similarity between any two neighbourhood graphs.\footnote{The Jaccard similarity is a widely used symmetric measure of the likeness of two sets.}
Choosing appropriate length minhash signatures squeezes the graph into memory and addresses problem~(1). To solve problem~(2) and achieve real-time querying we use the elements of the minhash signatures as the basis to build a Locality Sensitive Hashing (LSH) data structure. LSH facilitates querying of similar accounts in constant time.  This combination of minhashing and LSH allows analysts to enter an account or a set of accounts and in milliseconds receive the set of most related accounts. From this set we use the minhash signatures to rapidly construct a weighted graph and apply the WALKTRAP community detection algorithm before visualising the results~\citep{Pons2005}.

% What is innovative here
Our system applies well-studied techniques in an innovative way: (1) To the best of our knowledge, minhashing has not been applied to the neighbourhood graph before; Minhashing is normally only used for very similar sets. (2) We show that minhashing is effective for community detection when applied to a broad range of neighbourhood graph similarities. (3) We develop an agglomerative clustering algorithm and prove an original update procedure for minhash signatures in this setting. The novel combination of these techniques allows our system to perform real-time community detection on graphs that exceed 100 million vertices.

% our contribution
The contributions of this article are:
\begin{enumerate} 
\item We establish that robust associations between social media users can be determined by means of the Jaccard similarity of their neighbourhood graphs. 
\item We show that the approximations implicit in minhashing and LSH minimally degrade performance and allow querying of very large graphs in real time.  
\item System design and evaluation: We have designed and evaluated an end-to-end Python system for extracting data from social media providers, compressing the data into a form where it can be efficiently queried in real time.
\item We demonstrate how queries can be applied to a range of problems in graph analysis, e.g., understanding the structure of industries, allegiances within political parties and the public image of a brand.
\end{enumerate}

There are seven sections in this paper. Section 2 describes how to mine the Twitter graph and can be omitted by readers uninterested in replicating our work. Section 3 describes the related work, which is necessarily broad as our system brings together community detection, graph processing and data structures. Section 4 contains our detailed methodology with the exception of how we prepare and analyse the ground truth data, which is left until Section 5. In Section 6 we describe the results of three experiments, which validate our methodology and conclusions and future work follow in Section 7.

\section{Data and Preliminaries}
\label{sec:prelim}
% data sets intro
In this article, we focus on Twitter data because Twitter is the most widely used Digital Social Network (DSN) for academic research. The Twitter Follower graph consists of roughly one billion vertices (Twitter accounts) and 30 billion edges (Follows).

To show that our method generalises to other social networks, we also present some results using a Facebook Pages engagement graph containing 450 million vertices (FB accounts) and 700 million edges (Page likes / comments) (see Section~\ref{sec:evaluation}). 

% access tokens
Most DSNs have public Application Programming Interfaces (APIs) so that third-party developers can build applications using their data. Delivering data at massive scale incurs significant cost and to manage these, DSNs limit the rate that data can be downloaded. Rate limiting varies between networks. Usually, when a DSN account holder logs into a third party application using their social login, they grant the application owner one access token. Each access token allows the application owner to download a fixed amount of data in a given time window. This procedure gives more popular apps access to more data. Our work makes use of access tokens generated by several client facing apps\footnote{Starcount Playlist, Starcount Vibe and Chatsnacks}. 

% getting data from Twitter
To collect Twitter data we use the REST API to crawl the network identifying every account with more than 10,000 Followers\footnote{The number of Followers is contained in the Twitter account metadata, i.e.,  it is available without collecting and counting all edges.} and gather their complete Follower lists. Our data set contains 675,000 such accounts with a total of $1.5\times 10^{10}$ Followers, of which $7\times 10^8$ were unique. 
% Why only 675k users
We use accounts with greater than 10,000 Followers (though 700 million Twitter accounts are used to build the signatures) because accounts of this size tend to have public profiles (Wikipedia pages or Google hits) making the results interpretable. 
% Facebook data
To generate data from Facebook we matched the Twitter accounts with greater than 10,000 Followers to Facebook Page accounts \footnote{Facebook pages are the public equivalent of the private profiles. Many influential users have a Facebook Page.}  using a combination of automatic account name matching and manual verification. Facebook Page likes are not available retrospectively, but can be collected through a real-time stream. We collected the stream over a period of two years, starting in late 2013.
% on crawling social networks
Downloading large quantities of social media data is an involved subject and we include details of how we did this in Appedix~\ref{app:crawling} for reproducibility.

\section{Related Work}
\label{sec:rel}

% intro
Existing approaches to large scale, efficient, community detection have three flavours: More efficient community detection algorithms, innovative ways to perform processing on large graphs and data structures for graph compression and search. Table~\ref{tab:related_work} shows related approaches to this problem and which constraints they satisfy.
%\todo[inline]{A Ghahramani diagram to visualize Table 1 would be great (in addition, not instead).}
\begin{table}[htbp]
  \centering
  \caption{Comparison of related work. SCM stands for runs on a Single Commodity Machine}
    \begin{tabular}{l|ccc}
    \textbf{Method} & \textbf{Real-time} & \textbf{Large graphs} & \textbf{SCM} \\
    \hline
    Modularity optimisation \citep{Newman2004b}& \xmark & \xmark & \cmark    \\
    WALKTRAP \citep{Pons2005}& \xmark & \xmark & \cmark   \\
    INFOMAP \citep{Rosvall2008} & \xmark & \xmark & \cmark   \\
 Louvain method \citep{Blondel2008a} & \xmark & \cmark & \cmark  \\
 BigClam \citep{Yang2013}& \xmark & \cmark & \cmark   \\
 Graphlab \citep{Low2014Graphlab:Learning} & \xmark & \cmark & \xmark   \\
 Pregel \citep{Malewicz2010} & \xmark & \cmark & \xmark  \\
 Surfer \citep{Chen2010}  & \xmark & \cmark & \xmark  \\ 
 Graphci \citep{Kyrola2012} & \xmark & \cmark & \cmark  \\
 Twitter WTF \citep{Gupta2013}& \cmark & \cmark & \xmark   \\
 LEMON \citep{Li2015a} & \xmark & \cmark & \cmark  \\
 Our Method & \cmark & \cmark & \cmark 
    \end{tabular}%
  \label{tab:related_work}%
%  \figspace
\end{table}%

\subsection{Community Detection Algorithms}
% intro to community detection methods
% \todo[inline]{I could still add labelprop, spectral methods, stochastic blockmodel - the list is very long though}
Community detection methods have been developed in areas as diverse as neuronal firing \citep{Bullmore2009}, electron spin alignment \citep{Reichardt2006} and social models \citep{Yang2013}.  \citet{fortunato2010community} and \citet{Newman2003} both provide excellent and detailed overviews of the vast community detection literature. Approaches can be broadly categorised into local and global methods.

% Some current approaches to community detection; global methods - Walktrap, modularity, infomap
Global methods assign every vertex to a community, usually by partitioning the vertices. Many highly innovative schemes have been developed to do this. Modularity optimisation \citep{Newman2004b} is one of the best known. Modularity is a metric used to evaluate the quality of a graph partition. Communities are determined by selecting the partition that maximises the modularity. An alternative to modularity was developed by \cite{Pons2005} who innovatively applied random walks on the graph to define communities as regions in which walkers become trapped (WALKTRAP). \cite{Rosvall2008} combined random walks with efficient coding theory to produce INFOMAP, a technique that provides a new perspective on community detection: Communities are defined as the structural sub-units that facilitate the most efficient encoding of information flows through a network. All three methods are well optimised for their motivating networks, but were not created with graphs at the scale of modern Digital Social Networks (DSNs) and can not easily scale to very large data sets. 

% Some current methods that scale
The availability of data from the Web, DSNs and services like Wikipedia has focussed research attention on methods that scale. An early success was the Louvain method that allowed modularity optimisation to be applied to large graphs (they report 100 million vertices and 1 billion edges). However, the method was not intended to be real-time and the 152 minute runtime is too slow to achieve real-time performance, even allowing for 8 years of hardware advances \cite{Blondel2008a}. Another noteworthy technique applied to very large graphs is the Bigclam method, which in addition to operating at scale, is able to detect overlapping communities \citep{Yang2013}. However, in common with the Louvain method,  Bigclam is not a real-time algorithm that could facilitate interactive exploration of social networks. 

% Local methods 
In contrast to global community detection methods, local methods do not assign every vertex to a community. Instead they find vertices that are in the same community as a set of input vertices (seeds). For this reason they are normally faster than global methods. 
Local community detection methods were originally developed as crawling strategies to cope with the rapidly expanding
% \todo[inline]{The efficient PPR paper should be in here}
web-graph \citep{Flake2000}. Following the huge impact of the PageRank algorithm \citep{page1998a}, many local random walk algorithms have been developed. \cite{Kloumann2014} conducted a comprehensive assessment of local community detection algorithms on large graphs. In their study Personal PageRank (PPR)~\citep{Haveliwala2002} was the clear
winner. PPR is able to measure the similarity to a set of vertices instead of the global importance\slash influence of each vertex by applying a slight modification to PageRank. PageRank can be regarded as a sequence of two step processes that are iterated until convergence: A random walk on the graph followed by (with small probability) a random teleport to any vertex. PPR modifies PageRank in two ways: Only a small number of steps are run (often 4), and any random walker selected to teleport must return to one of the seed vertices. Recent extensions have shown that seeding PPR with the neighbourhood graph can improve performance \cite{Gleich2012} and that PPR can be used to initiate \emph{local} spectral methods with good results \cite{Li2015a}.

% evaluating local community methods.
Random walk methods are usually evaluated by power iteration; a series of matrix multiplications requiring the full adjacency matrix to be read into memory. The adjacency matrix of large graphs will not fit in memory and so distributed computing resources are used (e.g., Hadoop). While distributed systems are continually improving, they are not always available to analysts, require skilled operators and typically have an overhead of several minutes per query. %\todo[inline, color=yellow]{and require a network connection, presumably}

% verifying communities
A major challenge when applying both local and global community detection algorithms to real world networks is  performance verification. Testing algorithms on a held out labelled test set is complicated by the lack of any agreed definition of a community. Much early work makes use of small hand-labelled communities and treats the original researchers' decisions as gold standards \citep{sampson1969crisis, Zachary1977, Lusseau2003}. Irrespective of the validity of this process, a single (or small number) of manual labellers can not produce ground-truth for large DSNs. \cite{Yang2012} proposed a solution to the verification problem in community detection. They observe that in practice, community detection algorithms detect communities based on the structure of interconnections, but results are verified by discovering common attributes or functions of vertices within a community. \cite{Yang2012} identified 230 real-world networks in which they define ground-truth communities based on vertex attributes. The specific attributes that they use are varied and some examples include publication venues for academic co-authorship networks, chat group membership within social networks and product categories in co-purchasing networks.

\subsection{Graph Processing Systems}

A complimentary approach to efficient community detection on large graphs is to develop more efficient and robust systems. This is an area of active research within the systems community. General-purpose tools for distributed computation on large scale graphs include Graphlab, Pregel and Surfer \citep{Chen2010,Malewicz2010,Low2014Graphlab:Learning}. Purpose-built distributed graph processing systems offer major advances over the widely used MapReduce framework \citep{Pace2012}. This is particularly true for iterative computations, which are common in graph processing and include random walk algorithms. However, distributed graph processing still presents major design, usability and latency challenges. Typically the run times of algorithms are dominated by communication between machines over the network. Much of the complexity comes from partitioning the graph to minimise network traffic. The general solution to the graph partitioning problem is NP-hard and remains unsolved. These concerns have lead us and other researchers to buck the overarching trend for increased parallelisation on ever larger computing clusters and search for single-machine graph processing solutions. 
One such solution is Graphci, a single-machine system that offers a powerful and efficient alternative to processing on large graphs~\cite{Kyrola2012}. The key idea is to store the graph on disk and optimise I/O routines for graph analysis operations. Graphci achieves dramatic speed-ups compared to conventional systems, but the repeated disk I/O makes real-time operation impossible. Twitter also use a single-machine recommendation system that serves ``Who To Follow (WTF)'' recommendations across their entire user base~\citep{Gupta2013}. WTF provides real-time recommendations using random walk methods similar to PPR. They achieve this by loading the entire Twitter graph into memory. Following their design specification of 5 bytes per edge $5 \times 30 \times 10^9 = 150$ GB of RAM would be required to load the current graph, which is an order of magnitude more than available on our target platforms.

\subsection{Graph Compression and Data Structures}
% graph compression
The alternative to using large servers, clusters or disk storage for processing large graphs is to compress the whole graph to fit into the  memory of a single machine. Graph compression techniques were originally motivated by the desire for single machine processing on the Web Graph. Approaches focus on ways to store the differences between graph structures instead of the raw graph. \cite{Adler2001} searched for web pages with similar neighbourhood graphs and encoded only the differences between edge lists. The seminal work by \cite{Boldi2004} ordered Web pages lexicographically endowing them with a measure of locality. Similar compression techniques were adapted to social networks by \cite{Chierichetti2009}. They replaced the lexical ordering with an ordering based on a single minhash value of the out-edges, but found social networks to be less compressible than the Web (14 versus 3 bits per edge). While the aforementioned techniques achieve remarkable compression levels, the cost is slower access to the data \citep{Gupta2013}.

% minhashing 
Minhashing is a technique for representing large sets with fixed length signatures that encode an estimate of the similarity between the original sets. When the sets are sub-graphs minhashing can be used for lossy graph compression. The pioneering work on minhashing was by \cite{broder1997} whose implementation dealt with binary vectors. This was extended to counts (integer vectors) by \cite{Charikar2002} and later to continuous variables \citep{Philbin2008}. Efficient algorithms for
generating the hashes are discussed by \cite{Manasse2010}. 
% applications of minhashing
Minhashing has been applied to clustering the Web by \cite{Haveliwala2000}, who considered each web page to be a bag of words and built hashes from the count vectors. 

% advances on minashing
Two important innovations that improve upon minhashing are b-Bit minhashing \citep{Li2010} and Odd Sketches \citep{Mitzenmacher2014}. When designing a minhashing scheme there is a trade off between the size of the signatures and the variance of the similarity estimator. \cite{Li2010} show that it is possible to improve on the size-variance trade off by using longer signatures, but only keeping the lowest b-bits of each element (instead of all 32 or 64). Their work delivers large improvements for very similar sets (more than half of the total elements are shared) and for sets that are large relative to the number of elements in the sample space.  \cite{Mitzenmacher2014} improved upon b-bit minhashing by showing that for approximately identical sets (Jaccard similarities $\approx 1$) there was a more optimal estimation scheme. 

% locality sensitive hashing
Locality Sensitive Hashing (LSH) is a technique introduced by \citet{Indyk1998} for rapidly finding approximate near neighbours in high dimensional space. In the original paper they define a parameter $\rho$ that governs the quality of LSH algorithms. A lower value of $\rho$ leads to a better algorithm. There has been a great deal of work studying the limits on $\rho$. Of particular interest, \citet{Motwani2005} used a Fourier analytic argument to provide a tighter lower bound on $\rho$, which was later bettered by \cite{ODonnell2009} who exploited properties of the noise stability of boolean functions. The latest LSH research uses the structure of the data, through data dependent hash functions \cite{Andoni2014} to get even tighter bounds. As the hash functions are data dependent, unlike earlier work, only static data structures can be addressed. 

\section{Real-Time Community Detection}
\label{sec:method}

\begin{table}
\centering
\begin{center}
\begin{tabular}{ c c c }
\toprule
 System & Typical runtime (s) & Space requirement (GB) \\ 
\midrule
 Naive edge list & 8,000 & 240 \\  
 Minhash signatures & 1 & 4 \\
 LSH with minhash & 0.25 & 5 \\
\bottomrule
\end{tabular}
\end{center}
\caption{Typical runtimes and space requirements for systems performing local community detection on the Twitter Follower network of 700 million vertices and 20 billion edges and producing 100 vertex output communities}
\label{table:system_specs}
\end{table}

In this section, we detail our approach to real-time community detection in large social networks. 
%high-level overview of the method right here. What are the key steps? - just after the equation}
Our method consists of two main stages: In stage one, we take a set of seed accounts and expand this set to a larger group containing the most related accounts to the seeds. This stage is depicted by the box labelled "Find similar accounts" in Figure~\ref{fig:overview}. Stage one uses a very fast nearest neighbour search. In stage two, we embed the results of stage one into a weighted graph where each edge is weighted by the Jaccard similarity of the two accounts it connects. We apply a global community detection algorithm to the weighted graph and visualise the results. Stage two is depicted by the box labelled "Structure and visualise" in Figure~\ref{fig:overview}.

% Some notation
In the remainder of the paper we use the following notation: The $i^{th}$ user account (or interchangeably, vertex of the network) is denoted by $A_i$ and $N(A_i)$ gives the set of all accounts directly connected to  $A_i$ (the neighbours of $A_i$). The set of accounts that are input by a user into the system are called seeds and denoted by $S = \{A_1, A_2, ... ,A_m\}$ while $C = \{A_1, A_2, ... ,A_n \}$  (community) is used for the set of accounts that are returned by stage one of the process.

\subsection{Stage 1: Seed Expansion}

The first stage of the process takes a set of seed accounts as input, orders all other accounts by similarity to the seeds and returns an expanded set of accounts similar to the seed account(s). For this purpose, we require three ingredients:
\begin{enumerate}
\item A similarity metric between accounts
\item An efficient system for finding similar accounts
\item A stopping criterion to determine the number of accounts to return
\end{enumerate}
In the following, we detail these three ingredients of our system, which will allow for real-time community detection in large social networks on a standard laptop.

\subsubsection{Similarity Metric}

The property of each account that we choose to compare is the neighbourhood graph. The neighbourhood graph is an attractive feature as it is not controlled by an individual, but by the (approximately) independent actions of large numbers of individuals.  The edge generation process in Digital Social Networks (DSNs) is very noisy producing graphs with many extraneous and missing edges. As an illustrative example, the pop stars Eminem and Rihanna have collaborated on four records and a stadium tour.\footnote{``Love the
  Way You Lie'' (2010), ``The Monster'' (2013), ``Numb'' (2012), and
  ``Love the Way You Lie (Part II)'' (2010), the Monster Tour (2014)}
Despite this clear association, Eminem is not one of Rihanna's 40
million Twitter followers. However, Rihanna and Eminem have a Jaccard similarity of 18\%, making Rihanna Eminem's 6$^{th}$ strongest connection. Using the neighbourhood graph as the unit of comparison between accounts mitigates against noise associated with the unpredictable actions of individuals. 
The metric that we use to compare two neighbourhood graphs is the Jaccard similarity. The Jaccard similarity has two attractive properties for this task. Firstly it is a normalised measure providing comparable results for sets that differ in size by orders of magnitude. Secondly minhashing can be used to provide an unbiased estimator of the Jaccard similarity that is both time and space efficient.
% Jaccards work for associating twitter accounts
The Jaccard similarity is given by
\begin{align} \label{eq:jaccard}
J(A_i,A_j) = \frac{|N(A_i)\cap N(A_j)|}{|N(A_i)\cup N(A_j)|}\,,
\end{align}
where $N(A_i)$ is the set of neighbours of $i^{th}$ account. 
% The Jaccard similarity is related to the Jaccard distance by
% $
% D(.) = 1 - J(.),
% $
% which satisfies the three requirements of a distance metric\footnote{Triangle inequality, symmetry and $D(x,y) \geq 0$ with equality if and only if $xminhashing=y$}. 

\subsubsection{Efficient Account Search}

To efficiently search for accounts that are similar to a set of seeds we represent every account as a minhash signature and use a Locality Sensitive Hashing (LSH) data structure based on the minhash signatures for approximate nearest neighbour search.

\subsubsection*{Rapid Jaccard Estimation via Minhash Signatures}
\label{sec:use minhash}

Computing the Jaccard similarities in~\eqref{eq:jaccard} is very expensive as each set can have up to $10^{8}$ members and calculating intersections is super-linear in the total number of members of the two sets being intersected. Multiple large intersection calculations can not be processed in real-time. There are two alternatives: either the Jaccard similarities can be pre-computed for all possible pairs of vertices, or they can be estimated. Using pre-computed values for $n=675,000$ would require caching $\frac{1}{2}n(n-1) \approx 2.5 \times 10^{11}$ floating point values, which is approximately 1TB and so not possible using commodity hardware. Therefore an estimation procedure is required.

% minhashing
The minhashing compression technique of~\cite{broder2000min} generates unbiased estimates of the Jaccard similarity in $O(K)$, where $K$ is the number of hash functions in the signature. Each hash function approximates a two step process: An independent permutation of the indices associated with each member of a set followed by taking the minimum value of the permuted indices. \cite{broder2000min} showed that the unbiased estimate $\hat{J}(A_i,A_j)$ of the Jaccard similarity $J(A_i,A_j)$ is attained by exploiting that 
$$J(A_i,A_j)= p(h_k(A_i)=h_k(A_j)) \quad \forall k=1,\dotsc,K$$
where $h_i$ are hash functions . This means the probability that any minhash function is equal for both sets is given by the Jaccard coefficient. We create a signature vector $H$, which is made of $K$ independent hashes $h_i$ and calculate the Monte-Carlo Jaccard estimate $\hat J$ as
\begin{align} \label{eq:min_est}
  &\hat{J}(A_i,A_j) =  I/K,
\end{align}
where we define
\begin{align}
  &I = \sum_{k=1}^{K}\delta(h_{k}(A_i),h_{k}(A_j))\,,\\
  &\delta(h_k(A_i),h_k(A_j)) =
  \begin{cases}
   1 & \text{if } h_k(A_i) = h_k(A_j) \\
   0       & \text{if } h_k(A_i) \neq h_k(A_j)
  \end{cases}.
\end{align}
%
% error bounds
As each $h_k$ is independent, $I \sim Bin(J(A_i,A_j),K)$. The estimator is fully efficient, i.e., the variance is given by the Cram\'er-Rao lower bound 
\begin{align} 
\label{eq:minhash_variance}
\text{var}(\hat{J}) = \frac{J(1-J)}{K},
\end{align}
where we have dropped the Jaccard arguments for brevity. Equation~\ref{eq:minhash_variance} shows that Jaccard coefficients can be approximated to arbitrary precision using minhash signatures with an estimation error that scales as $O(1 /\sqrt{K})$.

% memory and space improvements of minhashing
The memory requirement of minhash signatures is $Kn$ integers, and so can be configured to fit into memory and for $K=1000$ and $n=675,000$ is only $\approx 4GB$. 
In comparison to calculating Jaccard similarities of the largest 675,000 Twitter accounts with $\approx 4 \times 10^{10}$ neighbours minhashing reduces expected processing times by a factor of $10,000$ and storage space by a factor of $1000$.\footnote{Our method allows to add  new accounts quickly by simply calculating one additional minhash signature without needing to add the pairwise similarity to all other accounts.}

\subsubsection*{Efficient Generation of Minhash Signatures}
\label{sec:gen minhash}
\begin{algorithm*}
	\caption{Minhash signature generation}
    \begin{algorithmic}[1]
        \Require $M \gets \textup{number of Accounts}$
        \Require $K \gets \textup{size of signature}$
        \Require $N(Account) \gets \textup{All Neighbours}$		
        	\State $T \in \mathbb{N}^{M \times K} \gets \infty$          	
            \Comment{Initialise signature matrix to $\infty$}
            \State index $\gets 1$
        	\ForAll{Accounts}
            	\State $P \gets \textup{permute}(\textup{index}) \in \mathbb{N}^{1 \times K}$
                \Comment{Permute the Account index $K$ times}
                \ForAll{N(Account)}
                    \State $T[i] \gets \min(T[i], P)$ \Comment{Compute the element-wise minimum of the signature}
                \EndFor
            \State index = index + 1
            \EndFor        
        \State \Return $T$ \Comment{Return matrix of signatures}   
    \end{algorithmic}
\label{alg:minhash1}
\end{algorithm*}

% problem: hash values is also expensive
Minhash signatures allow for rapid estimation of the Jaccard similarities. However, care must be taken when implementing minhash generation. Calculation of the signatures is expensive: Algorithm~\ref{alg:minhash1} requires $O(NEK)$ computations, where $N$ is the number of neighbours, $E$ is the average out-degree of each neighbour and $K$ is the length of the signature. For our Twitter data these values are $N = 7\times 10^{8}$, $E = 10$, $K = 1,000$. A naive implementation can run for several days. We have an efficient implementation that takes one hour allowing signatures to be regenerated overnight without affecting operational use (See Appendix~\ref{app:minhash_generation}).

% Section - nearest neighbour search with LSH
\subsubsection*{Locality Sensitive Hashing (LSH)}
\label{sec:LSH}

Calculating Jaccard similarities based on minhash signatures instead of full adjacency lists provides tremendous benefits in both space and time complexity. However, finding near neighbours of the input seeds is an onerous task. For a set of 100 seeds and our Twitter data set, nearly 70 million minhash signature comparisons would need to be performed, which dominates the run time. Locality Sensitive Hashing (LSH) is an efficient system for finding approximate near neighbours \cite{Indyk1998}.

% LSH description
LSH works by partitioning the data space. Any two points that fall inside the same partition are regarded as similar. Multiple independent partitions are considered, which are invoked by a set of hash functions. 
LSH has an elegant formulation when combined with minhash signatures for near neighbour queries in Jaccard space. The minhash signatures are divided into bands containing fixed numbers of hash values and LSH exploits that similar minhash signatures are likely to have identical bands. An LSH table can then be constructed that points from each account to all accounts that have at least one identical minhash band. We apply LSH to every input seed independently to find all candidates that are `near' to at least one seed.
% our implementation
In our implementation, we use 500 bands, each containing two hashes. As most accounts share no neighbours, the LSH step dramatically reduces the number of candidate accounts and the algorithm runtime by a factor of roughly 100. LSH is essential for the real-time capability of our system.

% Stage 3 clustering
\subsubsection*{Sorting Similarities}
\label{sec:s3}
LSH produces a set of candidate accounts that are related to at least one of the input seeds. In general, we do not want every candidate returned by LSH. Therefore, we select the subset of candidates that are most associated with the whole seed set. We experimented with two sequential ranking schemes: Minhash Similarity (MS) and Agglomerative Clustering (AC). The rankings can best be understood through the Jaccard distance $D = 1 - J$, which is used to define the centre $\mathbf{X} \in [0,1]^M$ of any set of vertices. At each step AC and MS augment the results set $C$ with $A^*$ the closest account to $\mathbf{X}: A^* \notin C$. However, MS uses a constant value of $\mathbf{X}$ based on the input seeds while AC updates $\mathbf{X}$ after each step. Formally, the centre of the input vertices used for MS is defined by
\begin{align}
X_j(A_j, S) = \frac{1}{n}\sum_{i=1}^n D(A_j, S_i)\,,\qquad j  = 0,1,\dotsc,M.
\label{eq:comm centre}
\end{align} 

\begin{algorithm}
	\caption{Agglomerative Clustering Algorithm (AC)}
    \begin{algorithmic}[1]
        \Require Community $C$ of initial seeds
        %\Require all Influencers $U_S$
        \State Define candidate members $\bar A = LSH(C)$
		\Repeat{}
        \State Compute community centre $\mathbf{X}(\bar A, S)$, see~\eqref{eq:comm centre}
        \State Select next Account: 
        	$A^* = \arg\min_{A_i \notin \bar C} \mathbf{X}(A_i,S)$
        \State Grow community: $C_{t+1}\gets C_t \cup A^*,\quad \bar A\gets \bar A\backslash A^* $
		\Until{Stopping criteria is met}        
    \end{algorithmic}
\label{alg:community}
\end{algorithm}

At each iteration of Algorithm~\ref{alg:community} $C$ and $\mathbf{X}$ are updated by first setting $C = S$ and then adding the closest account given by
\begin{align}
A^* = \arg\min\nolimits_{i} X(A_i,C) \quad \forall A_i \notin C
\label{eq:next S opt}
\end{align}
leading to
$$C_{t+1} = C_t \cup A^*.$$
The new centre $X_{n+1}$ is most efficiently calculated using the recursive online update equation
\begin{align}
X_{t+1}(A,C_{t+1}) = \frac{nX(A,C_t) + D(A^*, C_t)}{n+1},
\end{align}
where $n$ is the size of $C_t$.

\subsubsection{Stopping Criterion}
\label{sec:stopping}
% stopping criteria
Both AC and MS are sequential processes and will return every candidate account unless a stopping criteria is applied.
Many stopping criteria have been used to terminate seed expansion
processes. The simplest method is to terminate after a fixed number of
inclusions. Alternative methods use local maxima in
modularity~\citep{Lancichinetti2009} and
conductance~\citep{Leskovec2010}. 

An application of our work is to help define an optimal set of celebrities to endorse a brand. In this context we want to answer questions like: ``What is the smallest set of athletes that have influence on over half of the users of Twitter?''. We refer to the number of unique neighbours of a set of accounts as the \emph{coverage} of that set. An exact solution to this problem is combinatorial and requires calculating large numbers of unions over very large sets. However it can be efficiently approximated using minhash signatures. We exploit two properties of minhash signatures to do this: The unbiased Jaccard estimate through Equation~\ref{eq:min_est} and the minhash signature of the union of two sets is the elementwise minimum of their respective minhash signatures.  Minhash signatures allow coverage to be used as a stopping criteria to rank LSH candidates without losing real-time performance.

\paragraph{Efficient Coverage Computation}
% calculating the total number of followers of the agglomerated community
The coverage $y$ is given by
\begin{align}
y = \left|\bigcup_{i=1}^n N(A_i)\right|\,,
\end{align}
the number of unique neighbours of the output vertices. Every time a new account $A$ is added we need to calculate $|N(C) \cup N(A)|$ to update the coverage. This is a large union operation and expensive to perform on each addition. Lemma~\ref{lemma} allows us to rephrase this expensive computation equivalently by using the Jaccard coefficient (available cheaply via the minhash signatures), which we subsequently use for a real-time iterative algorithm.

% Lemma
\begin{lemma}\label{lemma1} 
For a community $C = \bigcup_i A_i$ and a new account $A \notin C$, the number of Neighbours of the union $A \cup C$ is given as
\begin{align}
\left |N( A \cup C ) \right | =  \frac{| N(A)  | +  | N(C) | }{1+J\left (A,C  \right )}.
\label{eq:SuC}
\end{align}
\end{lemma}

\begin{proof}
Following~\eqref{eq:jaccard}, the Jaccard coefficient of a new Account $A\notin C$ and the community $C$ is
 \begin{equation} \label{eq:1}
J(A,C) = \frac{\left | A\cap C \right |}{\left | A\cup C \right |}.
\end{equation}
By considering the Venn diagram and utilising the inclusion-exclusion principle, we obtain
\begin{equation} \label{eq:2}
\left | A\cup C \right | = \left | A \right | + \left | C \right | -\left | A\cap C \right |.
\end{equation}
Substituting this expression in the denominator of the Jaccard coefficient in~\eqref{eq:1} yields
\begin{align*}
\frac{|A| + |C|}{1+J(A,C)} &\stackrel{\begin{subarray}{c}\eqref{eq:1}\\ \eqref{eq:2}\end{subarray}}{=} \frac{|A| + |C|}{1+\frac{|A\cap C|}{\left | A \right | + \left | C \right | -\left | A\cap C \right |}}=\frac{|A| + |C|}{\frac{|A|+|C|}{\left | A \right | + \left | C \right | -\left | A\cap C \right |}}\\
&= |A| + |C|-|A\cap C|\stackrel{\eqref{eq:2}}{=} |A\cup C|\,,
\end{align*}
which proves~\eqref{eq:SuC} and the Lemma.
\end{proof}

% Lemma
\begin{lemma}\label{lemma2} 
A community $C = \bigcup_i A_i$ can be represented by a minhash signature $H(C) = \{h_1(C), h_2(C), ... h_k(C)\}$ where
\begin{align}
\label{eq:mhu}
h_i(C) = min_j(h_i(A_j))
\end{align}
\end{lemma}

\begin{proof}
A minhash signature is composed of $k$ independent minhash functions. Each of which is a compound function made up of a general mapping and a minimum operation.
$$h_i(A) = \min(g_i(A)) \quad \forall i,A$$
where $g_i: \mathbb{N}^m \rightarrow \mathbb{N}^m$
and so 
\begin{align}
h_i(\cup_j A_j) &= \min(g_i(\cup_j A_j)) \\
h_i(C) &= \min(g_i(A_1), g_i(A_2), ...,g_i(A_k))
\end{align}
which proves~\eqref{eq:mhu} and the Lemma.
\end{proof}

% size of the new community is now easy
We use Lemma~\ref{lemma1} to update the unique neighbour count. Once the next account to add to the community $A^*$ is determined according to~\eqref{eq:next S opt}
\begin{align}
|N(C_{t+1})| =  \frac{|N(C_{t})| + |N(A^*)|}{1+J(C_{t},A^*)}.
\label{eq: lemma_res}
\end{align}
The right hand side of \eqref{eq: lemma_res} contains three terms: $|N(C_{t})|$ is what we started with, $|N(A^*)|$ is the neighbour count of $A^*$, which is easily obtained from Twitter or Facebook metadata and $J(C_{t},A^*)$ is a Jaccard calculation between a community and an account. The minhash signature of a community is obtained via \eqref{eq:mhu} and so we are able to calculate the coverage with negligible additional computational overhead.

\subsection{Stage 2: Community Detection and Visualisation}
% summarize stage 1
Stage one expanded the seed accounts to find the related region. This was done by first finding a large group of candidates using LSH that were related to any one of the seeds and then filtering down to the accounts most associated to the whole seed set. 

% stage 2
In Stage two, the vertices returned by Stage one are used to construct a weighted Jaccard similarity graph. Figure~\ref{fig:network_construction} depicts the process of transforming from the original unweighted graph to the weighted graph. The red vertices are those returned by stage one. Edge weights are calculated for all pairwise associations from the minhash signatures through Equation~\ref{eq:min_est}. This process effectively embeds the original graph in a metric Jaccard space~\citep{broder1997}. Community detection is run on the weighted graph. 

% 3 stages -> figure
\begin{figure*}
        \centering
\begin{subfigure}[b]{0.32\textwidth}
                \includegraphics[width=\textwidth]{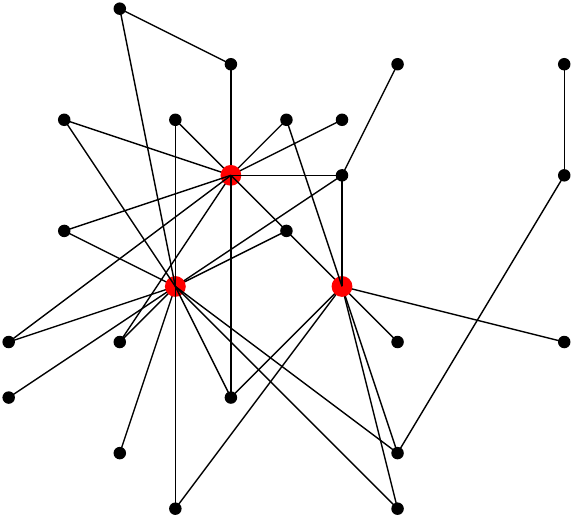}
                \caption{A social network containing three interesting (red) vertices.}
                \label{fig:full_graph}
        \end{subfigure}
        \hfill
        \begin{subfigure}[b]{0.32\textwidth}
                \includegraphics[width=\textwidth,clip]{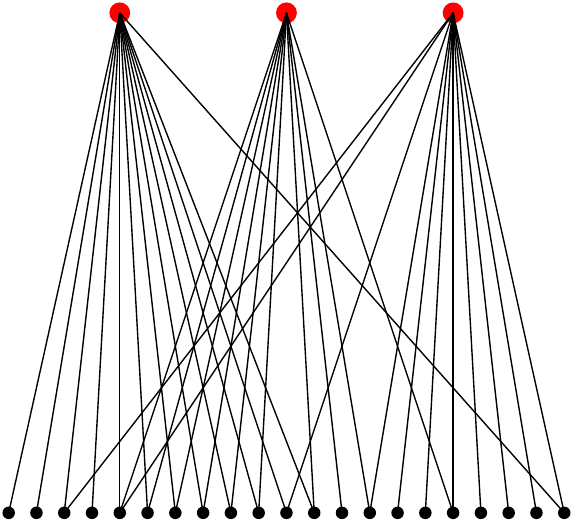}
                \caption{The overlapping neighbourhood graphs of the interesting vertices in (\subref{fig:full_graph})}
                \label{fig:bipartite_graph}
        \end{subfigure}     
      \hfill
        \begin{subfigure}[b]{0.32\textwidth}
                \includegraphics[width=\textwidth,clip]{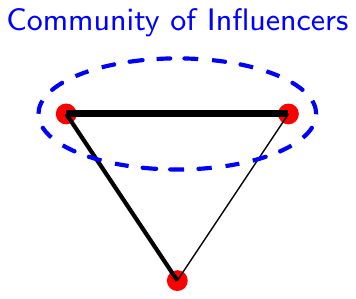}
                \caption{The inferred weighted network. Vertices connected by high weights are more likely to be in the same community}
                \label{fig:influencer_graph}
        \end{subfigure}  
%         \begin{subfigure}[b]{0.24\textwidth}
%                 \includegraphics[width=\textwidth,clip]{method/c4}
%                 \caption{Inferred community (full Twitter graph).}
%                 \label{fig:influencer_graph_community_detection}
%         \end{subfigure} 
\caption{Visualising the generation of a weighted subgraph. Interesting vertices are depicted as larger red nodes, and the neighbours as smaller, more numerous black nodes. (\subref{fig:full_graph}) shows a complete social network.  (\subref{fig:bipartite_graph}) depicts the overlapping neighbourhood graphs of the three interesting vertices in (\subref{fig:full_graph}). (\subref{fig:influencer_graph}) The network is summarised by an inferred network using the Jaccard similarity measure of the set of neighbouring vertices as edge weights.}
\label{fig:network_construction}
\end{figure*}

%\subsubsection{Visualisation}
% visualization
The final element of the process is to visualise the community structure and association strengths in the region of the input seeds. We experimented with several global community detection algorithms. These included INFOMAP, Label Propagation, various spectral methods and Modularity Maximisation \citep{Rosvall2008, Raghavan2007, Newman2006, Newman2004c}. The Jaccard similarity graph is weighted and almost fully connected and most community detection algorithms are designed for binary sparse graphs. As a result, all methods with the exception of label propagation and WALKTRAP were too slow for our use case. Label Propagation had a tendency to select a single giant cluster, thus adding no useful information. Therefore, we chose WALKTRAP for community visualisation.

\section{Ground-Truth Communities}

To provide a quantitative assessment of our method we require ground-truth labelled communities. No ground-truth exists for the data sets of interest and so in this section we provide a methodology for generating ground-truth. This methodology itself must be verified and we provide an extensive evaluation of the quality of the derived ground-truth based on the axiomatic definitions described in \cite{Yang2012}.
% The problem of verification - structural and functional groups
Most community detection algorithms (including ours) are based on the structure of the graph \citep{Fortunato2007}. Axiomatically, good community structures are:
\begin{itemize}
\item Compact
\item Densely interconnected
\item Well separated from the rest of the network
\item Internally homogeneous
\end{itemize}
However, while communities are detected using these properties, verification typically requires associating each vertex with some functional attributes, e.g., fans of Arsenal football club or Python programmers and showing that the discovered communities group attributes together \citep{Yang2012}. The practice of relating community membership with personal attributes is justified by the homophily principle of social networks \citep{McPherson2001}, which states that people with similar attributes are more likely to be connected.
% What we do with Wikipedia
We reverse the process of verification by generating ground-truth from personal attributes. To generate attributes we match Twitter accounts with Wikipedia pages and associate Wikipedia tags with each Twitter account. Wikipedia tags give  hierarchical functions like `football:sportsperson:sport' and `pop:musician:music'. It is not possible to match every Twitter account and our matching process discovered 127 tags that occur more than 100 times in the data.  Of these, many were clearly too vague to be useful such as `news:media' or `Product Brand:Brands'. We selected 16 tags that had relatively high frequencies in the data set and evaluated 7 metrics for each that are related to the four axioms. These result are shown in Table~\ref{tab:communities}. Seperability and conductance measure how well separated a community is from the rest of the graph. Density and size measure the compactness and density. Cohesiveness, clustering and conductance ratio measure how internally homogeneous a community is. The mathematical formulation of these metrics and details of how they were calculated is provided in Appendix~\ref{app:ground_truth}. 
% Description of the table
Table~\ref{tab:communities} is sorted by density and the bold rows are visualised in Figures~\ref{fig:taekwondo_network},\ref{fig:alcohol_network},\ref{fig:hotel_network} and \ref{fig:basketball_network}. The density is the most important factor distinguishing good from bad communities, varying by two orders of magnitude across the data. This is followed by how well separated (separability) the community is from the rest of the network, which is inversely correlated with conductance by design (See Equations~\ref{eq:con}~and~\ref{eq:sep}). High clustering is also a useful indicator of community goodness for the best communities, but is less useful for separating communities that are made up of many sub-units like team sports from very bad communities like Food and Drink. Cohesiveness is generally not useful as most communities contain at least one well separated sub-unit.

% The old table caption
% Mixed Martial Arts, Adult Actor and Cycling form low modality, densely interconnected communities that are well separated from the rest of the network and have high clustering. The next 4 communities (Baseball, Basketball, American Football and Athletics) are highly multi-modal, well separated from the rest of the network and contain densely intra-connected sub-communities. Football is the largest community and is extremely multi-modal. The remaining communities are weakly connected and contain disconnected sub-communities. 

\begin{table*}[t]
  \centering
  \caption{Properties of ground-truth communities sorted by edge density. CR stands for Conductance Ratio. High values of clustering, density and separability and low values of CR, conductance and Cohesiveness indicate good communities.}
  \scalebox{0.75}{
    \begin{tabular}{l|ccccccc}
    \small
    \textbf{Community} & \textbf{Size} & \textbf{Clustering} & \textbf{Cohesiveness} & \textbf{Conductance} & \textbf{CR} & \textbf{Density} & \textbf{Separability} \\
    \hline
    \textbf{Mixed Martial Arts} & \textbf{751}   & \textbf{6.49E-02} & \textbf{4.29E-01} & \textbf{5.10E-01} & \textbf{1.19} & \textbf{3.06E-02} & \textbf{4.80E-01} \\
    \textbf{Adult Actors} & 352   & 7.20E-02 & 1.29E-01 & 7.70E-01 & 5.98 & 2.94E-02 & 1.50E-01 \\
    \textbf{Cycling} & 371   & 6.43E-02 & 4.51E-01 & 7.04E-01 & 1.56 & 2.50E-02 & 2.11E-01 \\
    \textbf{Baseball} & 616   & 3.64E-02 & 1.49E-01 & 7.87E-01 & 5.29 & 1.63E-02 & 1.35E-01 \\
    \textbf{Basketball} & \textbf{786}   & \textbf{3.84E-02} & \textbf{3.30E-01} & \textbf{7.71E-01} & \textbf{2.34} & \textbf{1.60E-02} & \textbf{1.48E-01} \\
      \textbf{American Football} & 1295  & 2.24E-02 & 3.82E-01 & 7.40E-01 & 1.94 & 9.33E-03 & 1.75E-01 \\
            \textbf{Athletics} & 530  & 3.48E-02 & 4.13E-01 & 8.47E-01 & 2.05 & 8.21E-03 & 9.01E-02 \\
        \textbf{Hotel Brand} & \textbf{836}   & \textbf{2.20E-02} & \textbf{4.53E-01} & \textbf{8.37E-01} & \textbf{1.85} & \textbf{6.16E-03} & \textbf{9.71E-02} \\
    \textbf{Airline} & 363   & 2.30E-02 & 4.41E-01 & 9.46E-01 & 2.15 & 4.35E-03 & 2.84E-02 \\
    \textbf{Cosmetics} & 332   & 3.34E-02 & 4.87E-01 & 9.56E-01 & 1.96 & 3.55E-03 & 2.32E-02 \\
    \textbf{Football} & 4111  & 3.69E-02 & 3.95E-01 & 7.07E-01 & 1.79 & 2.93E-03 & 2.07E-01 \\
    \textbf{Alcohol} & \textbf{388}   & \textbf{1.72E-02} & \textbf{2.34E-01} & \textbf{9.52E-01} & \textbf{4.06} & \textbf{2.66E-03} & \textbf{2.53E-02} \\
    \textbf{Travel} & 2038  & 1.27E-02 & 4.25E-01 & 8.29E-01 & 1.95 & 2.50E-03 & 1.03E-01 \\
    \textbf{Model} & 2096  & 2.62E-02 & 4.04E-01 & 9.01E-01 & 2.23& 1.90E-03 & 5.50E-02 \\
    \textbf{Electronics} & 689 & 1.40E-02 & 4.38E-01 & 9.75E-01 & 2.23 & 8.78E-04 & 1.30E-03 \\
    \textbf{Food and Drink} & 2974  & 1.76E-02 & 4.57E-01 & 9.06E-01 & 1.98 & 7.69E-04 & 5.18E-02 \\
    \end{tabular}%
    }
  \label{tab:communities}%
\end{table*}%

% dendrograms
\begin{figure}
\begin{tabular}{cccc}
\begin{subfigure}[t]{0.23\textwidth}
\includegraphics[width=\hsize,clip]{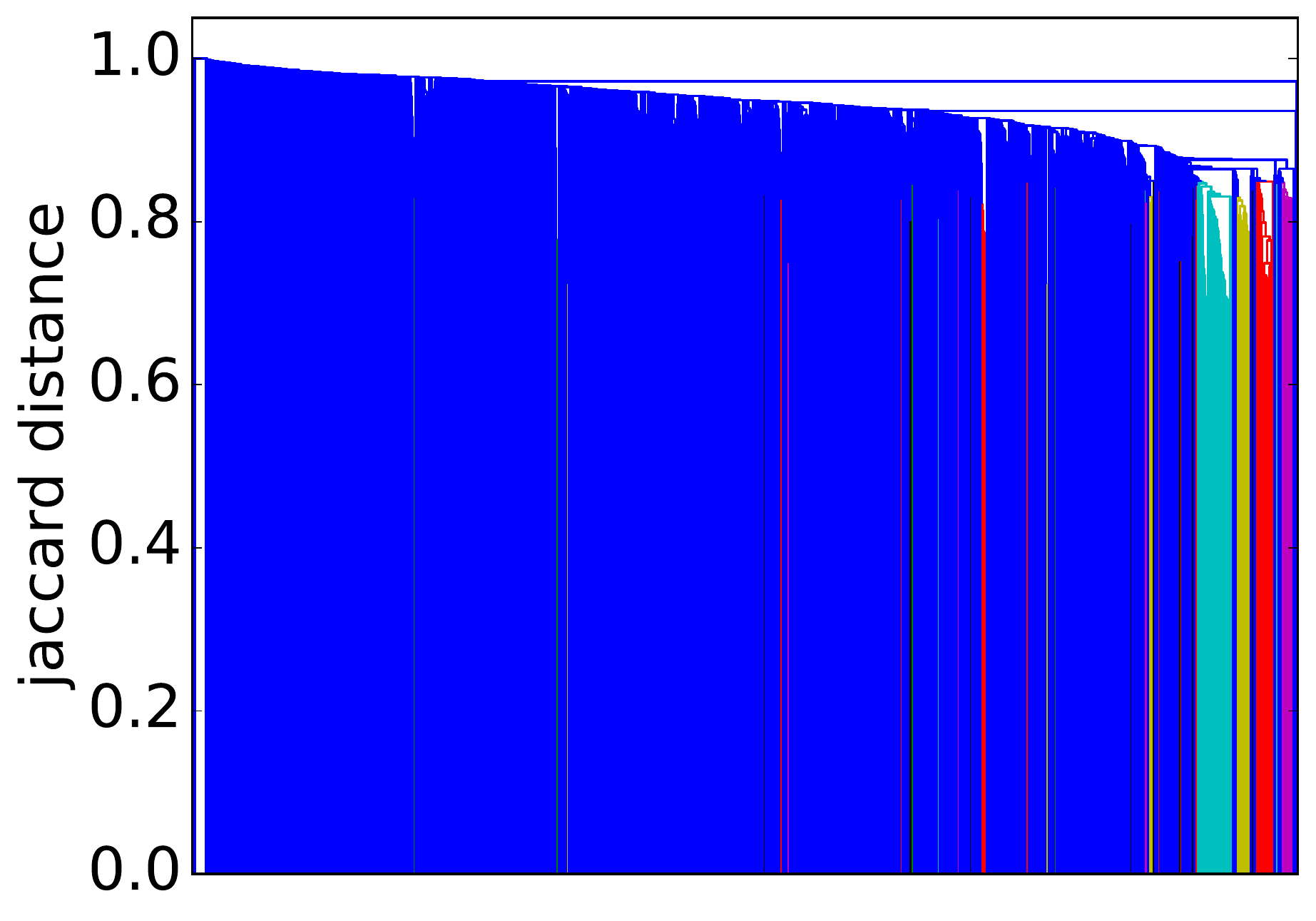}
\caption{Travel}
\label{fig:travel_dendro}  
\end{subfigure}
~
% airline dendro
\begin{subfigure}[t]{0.23\textwidth}
\includegraphics[width=\textwidth]{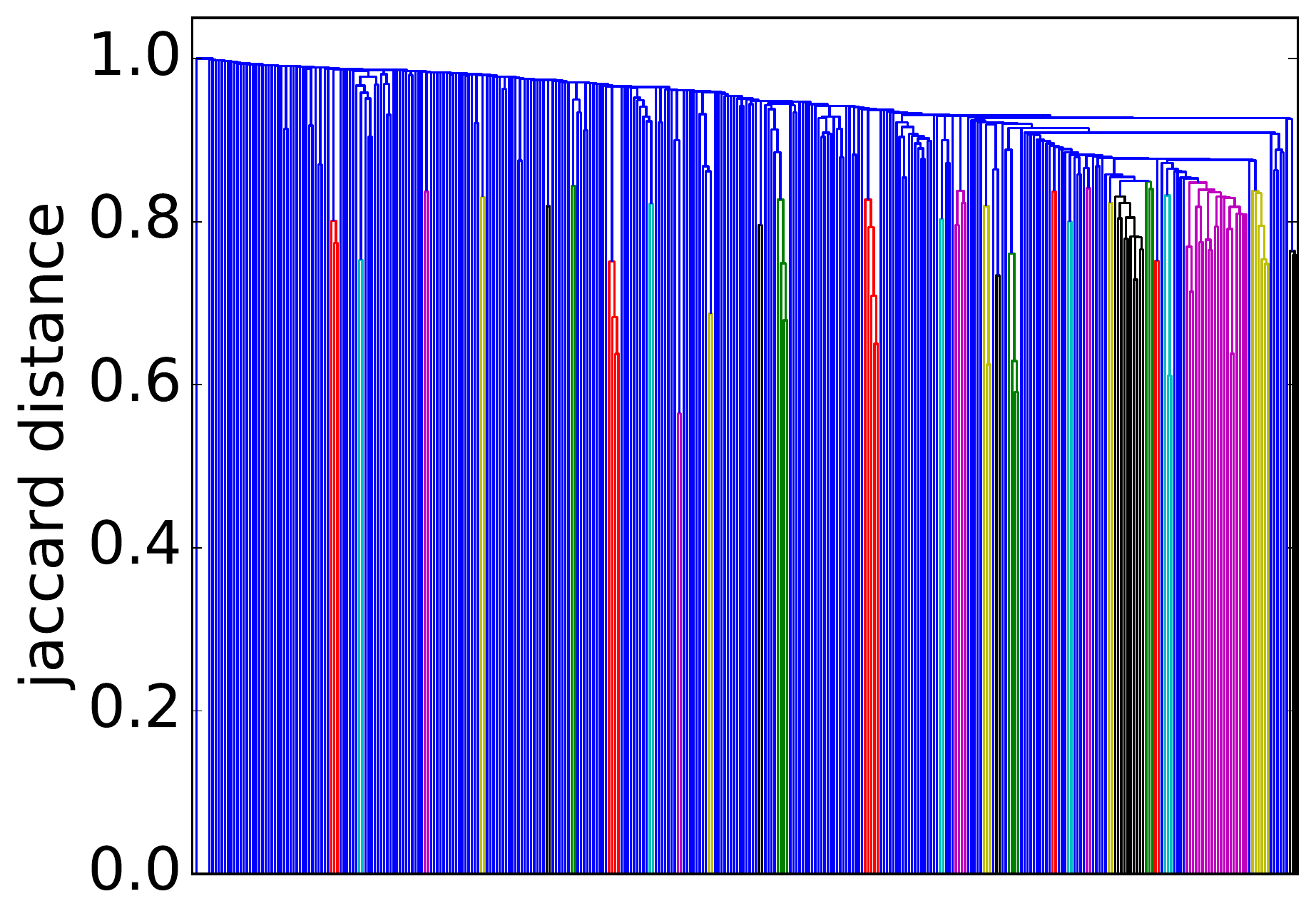}
\caption{Airline}
\label{fig:airline_dendro}
\end{subfigure}
~
% hotel dendro
\begin{subfigure}[t]{0.23\textwidth}
\includegraphics[width=\textwidth]{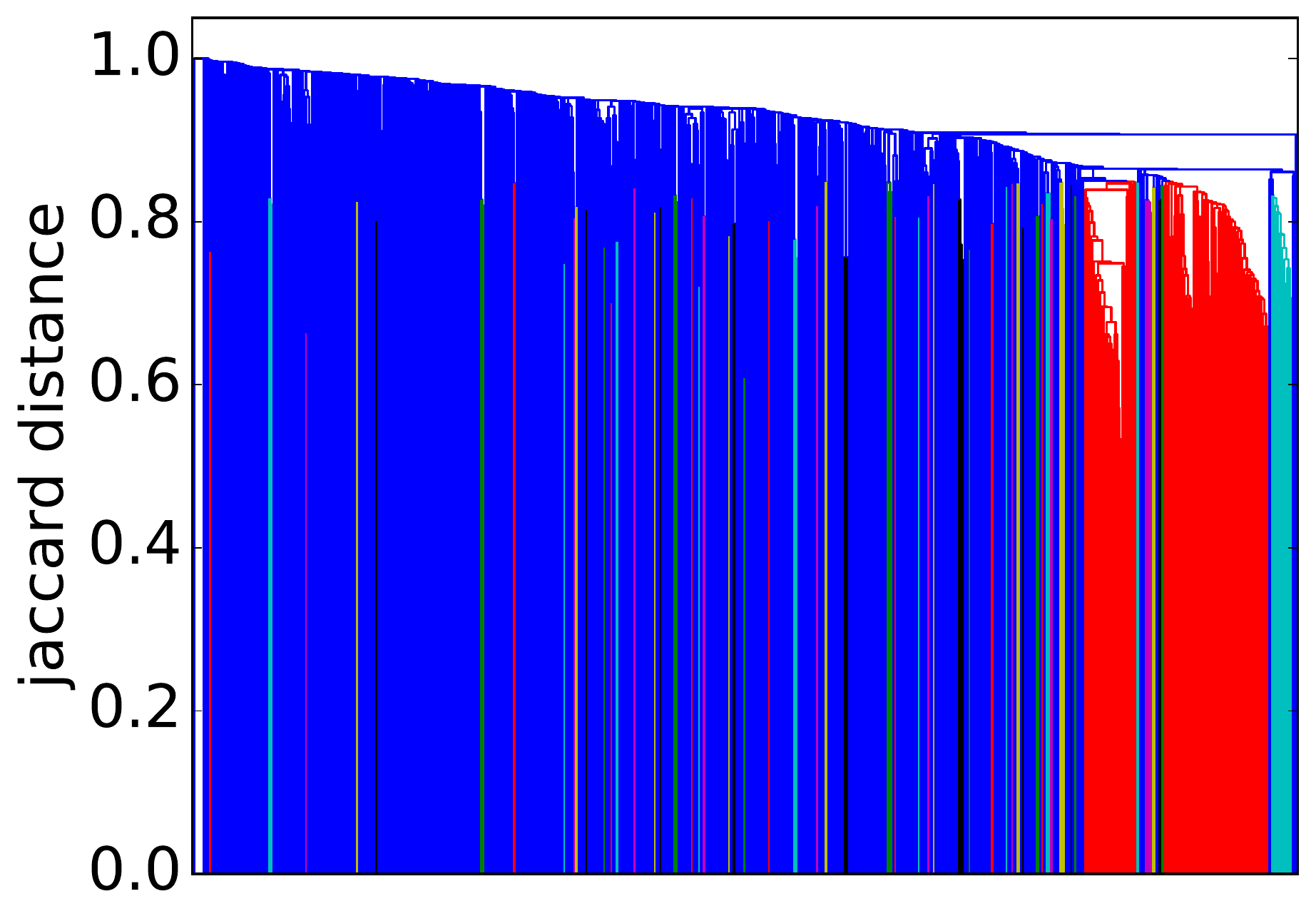}
\caption{Hotel}
\label{fig:hotel_dendro}
\end{subfigure}
~
% cosmetics dendro
\begin{subfigure}[t]{0.23\textwidth}
\includegraphics[width=\textwidth]{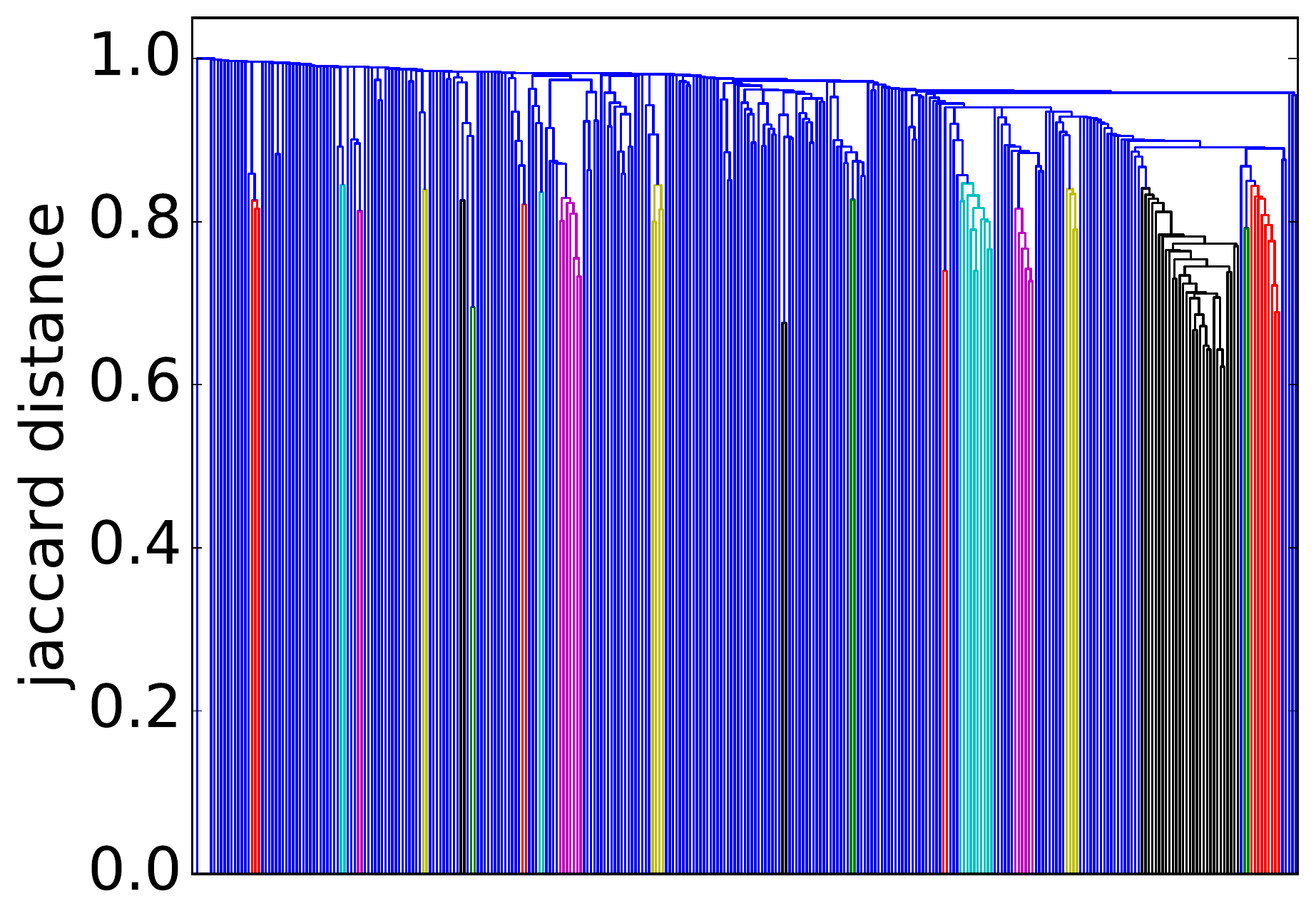}
\caption{Cosmetics}
\label{fig:cosmetics_dendro}
\end{subfigure} 
\\
Industrial groups. Small highly connected groups due to sub-brands 
\\
\begin{subfigure}[t]{0.23\textwidth}
\includegraphics[width=\hsize,clip]{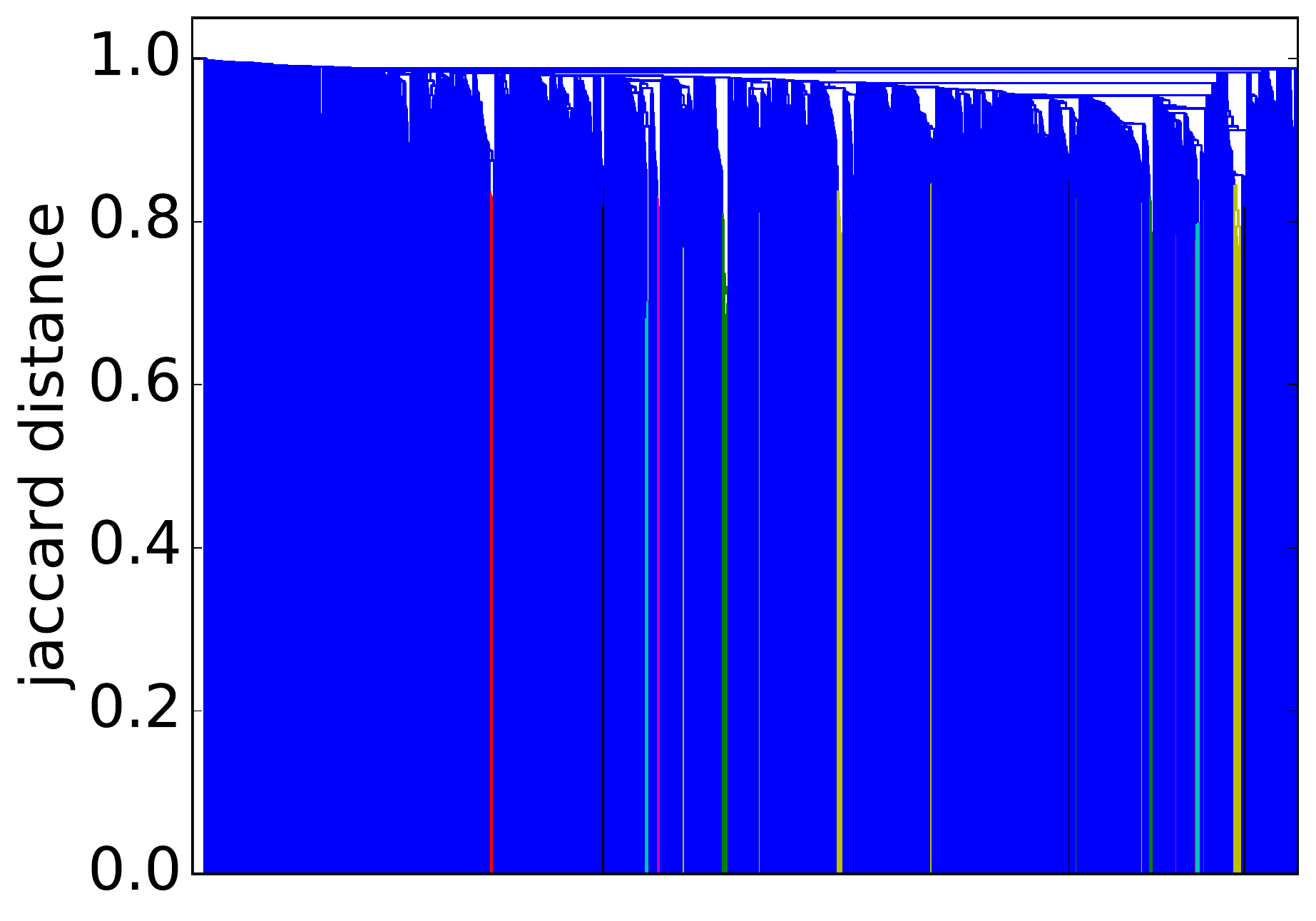}
\caption{Food and Drink}
\label{fig:food_drink_dendro}  
\end{subfigure}
~
% electronics dendro
\begin{subfigure}[t]{0.23\textwidth}
\includegraphics[width=\textwidth]{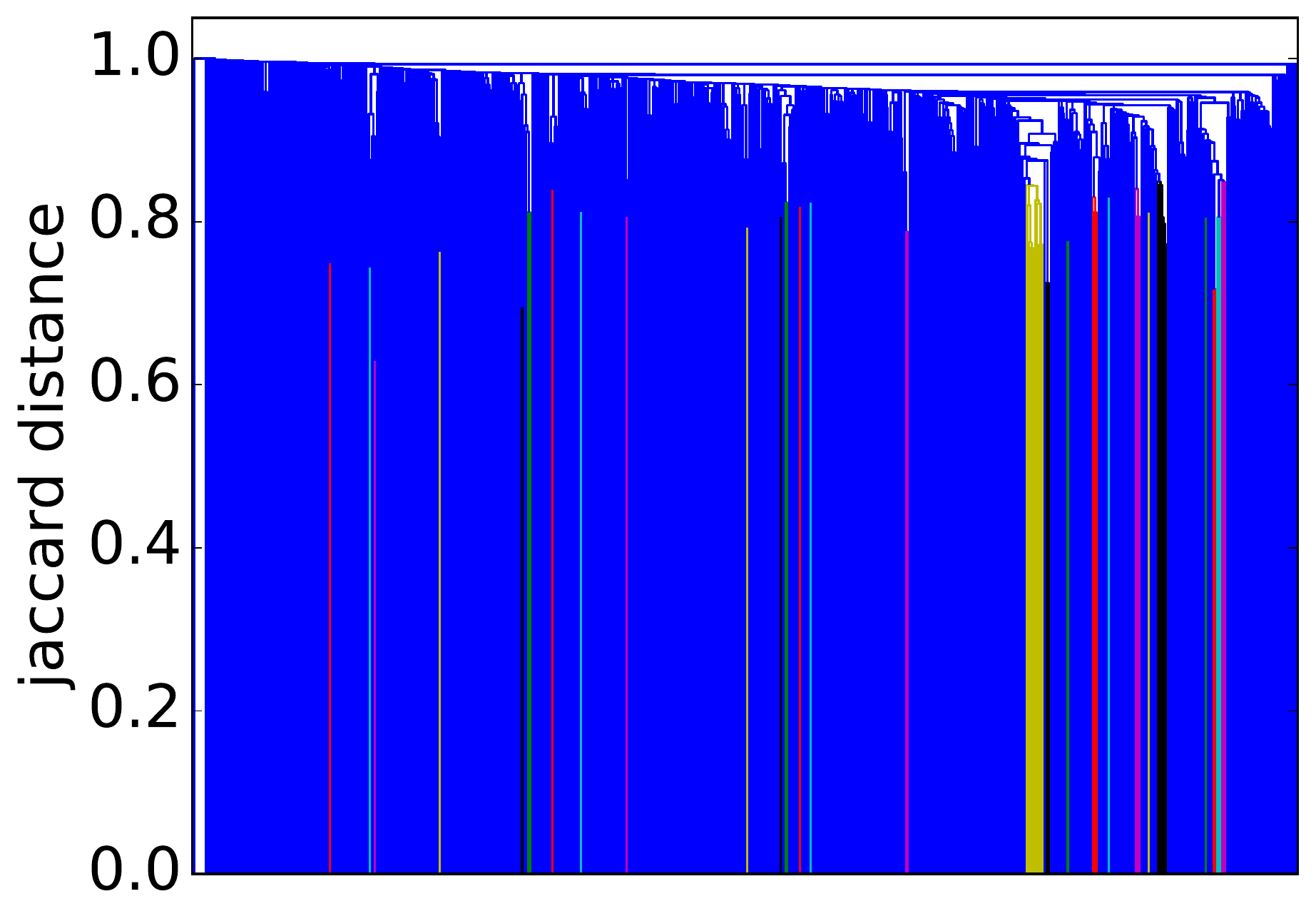}
\caption{Electronics}
\label{fig:electronics_dendro}
\end{subfigure}
~
% alcohol dendro
\begin{subfigure}[t]{0.23\textwidth}
\includegraphics[width=\textwidth]{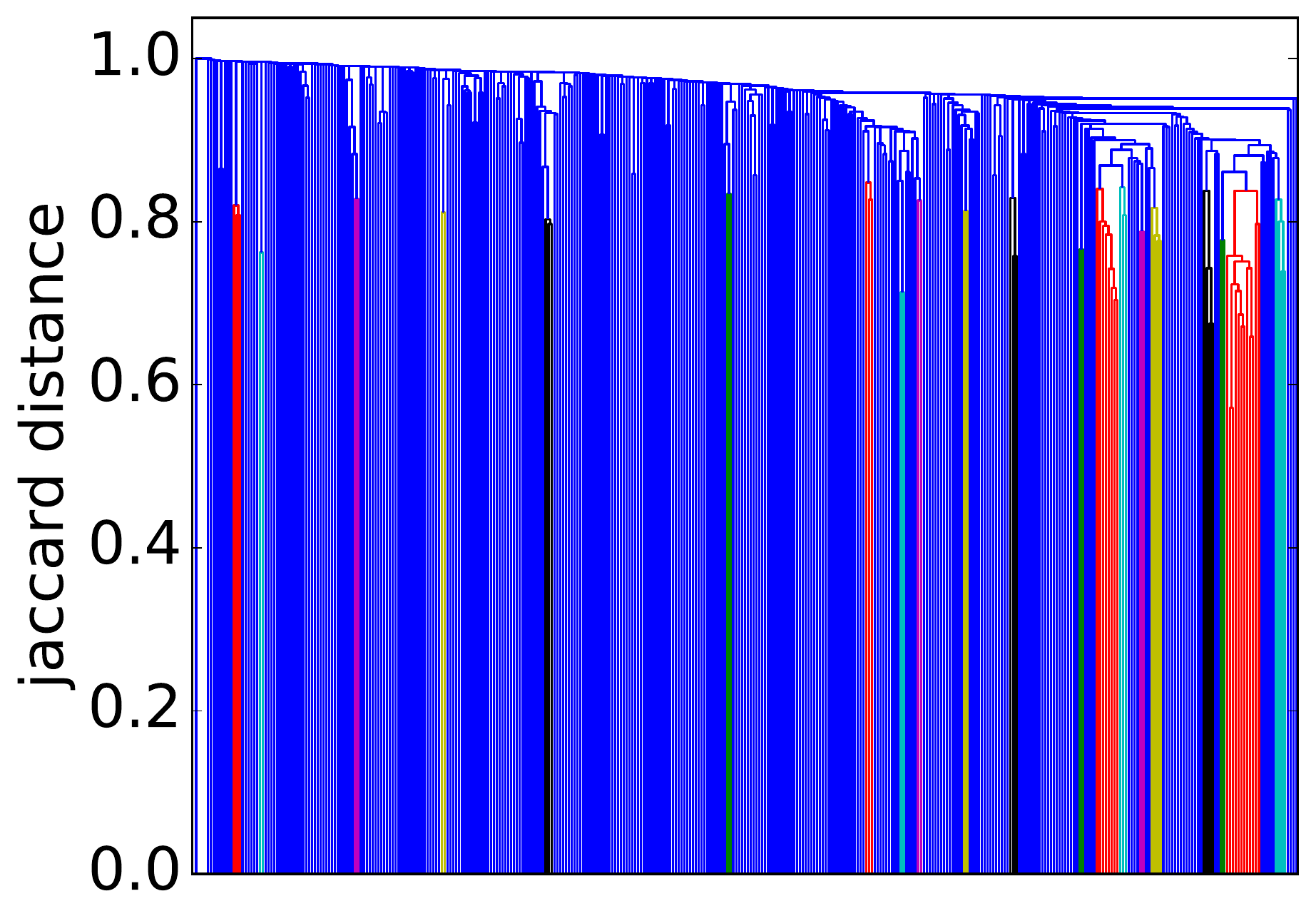}
\caption{Alcohol}
\label{fig:alcohol_dendro}
\end{subfigure}
~
% model dendro
\begin{subfigure}[t]{0.23\textwidth}
\includegraphics[width=\textwidth]{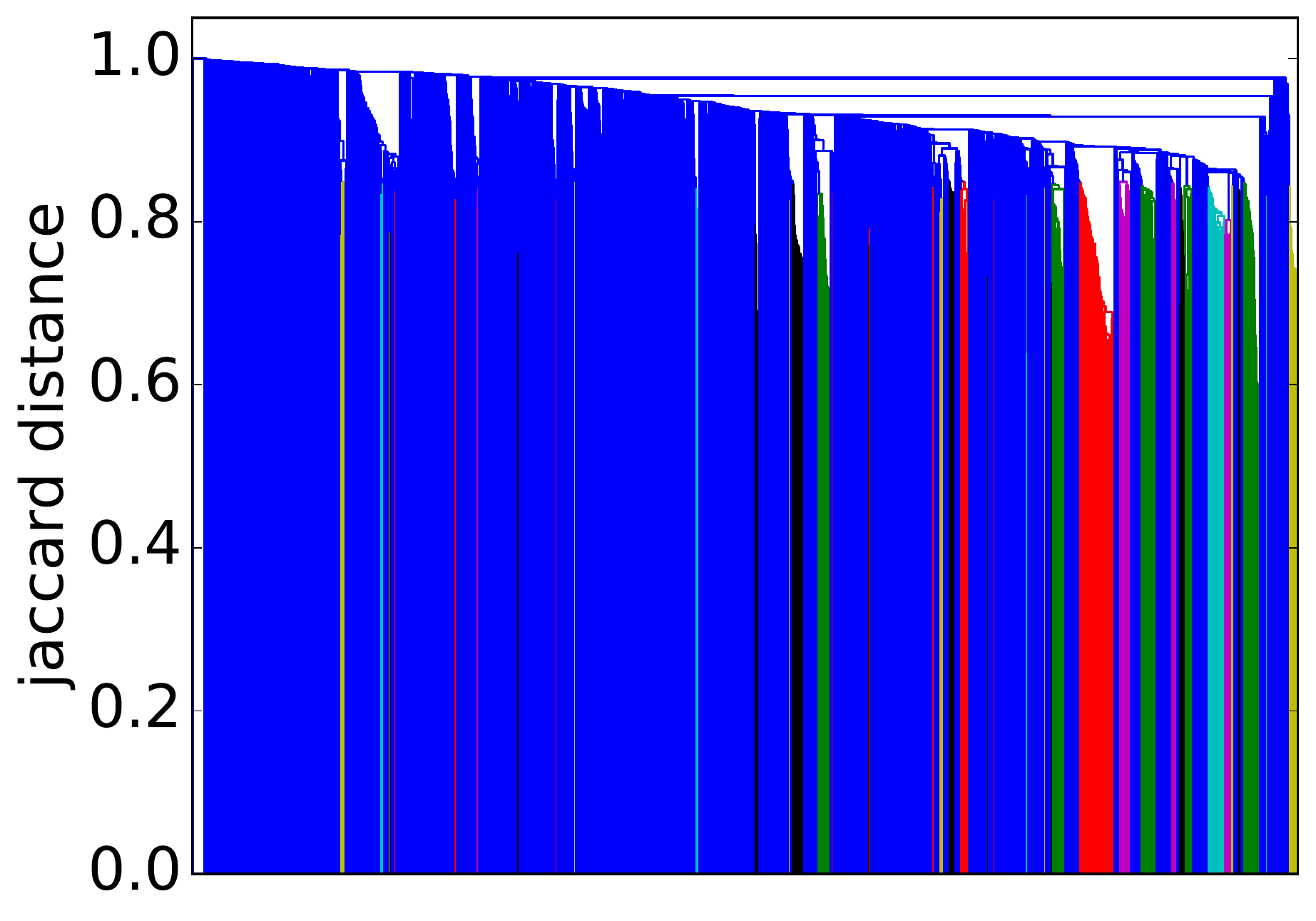}
\caption{Model}
\label{fig:model_dendro}
\end{subfigure}
\\
Industrial groups. Limited interaction
\\
%\centering
% MMA dendro
\begin{subfigure}[t]{0.23\textwidth}
\includegraphics[width=\textwidth]{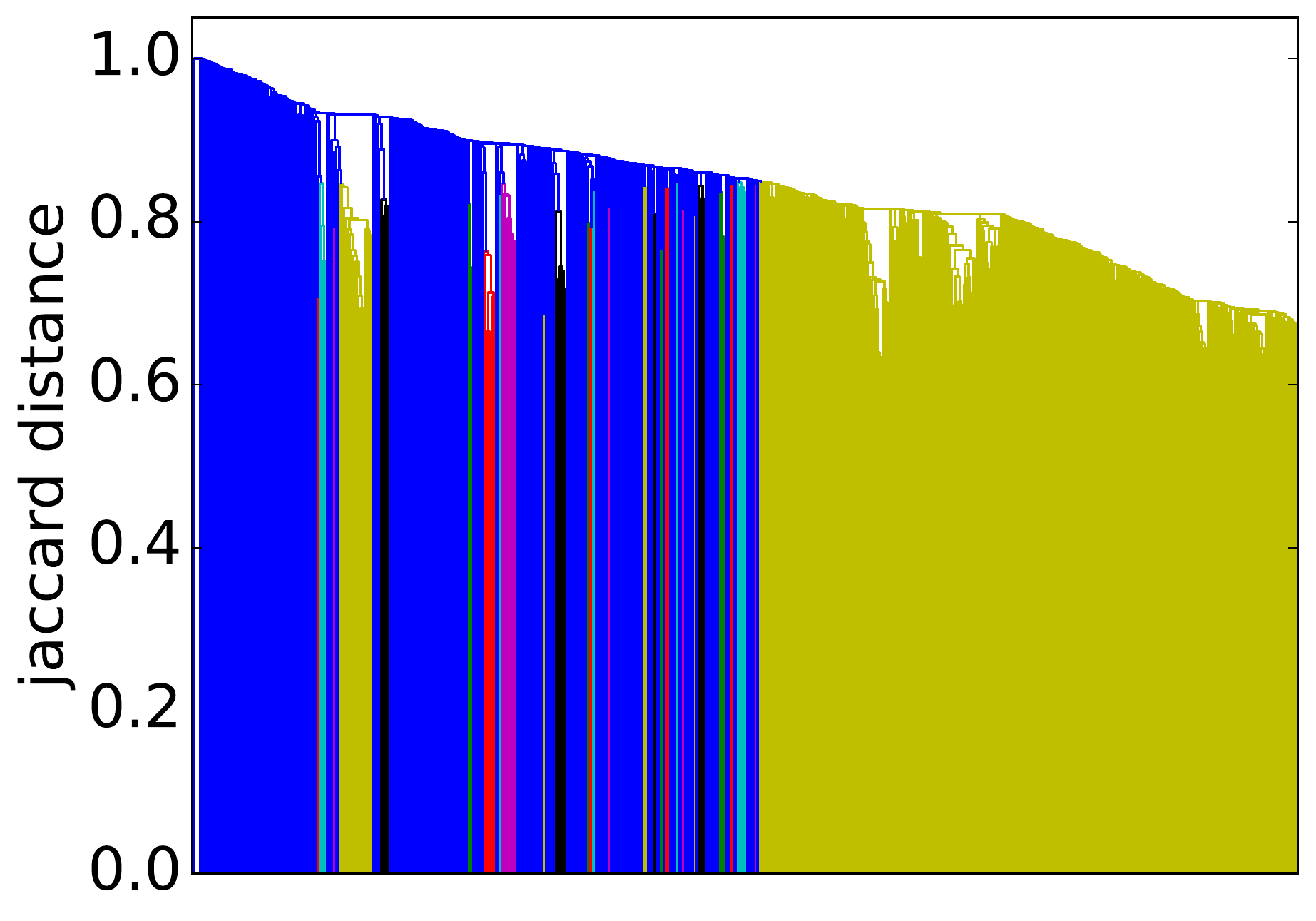}
\caption{Mixed Martial Arts}
\label{fig:taekwondo_dendro}
\end{subfigure} 
~
% cycling dendro
\begin{subfigure}[t]{0.23\textwidth}
\includegraphics[width=\hsize,clip]{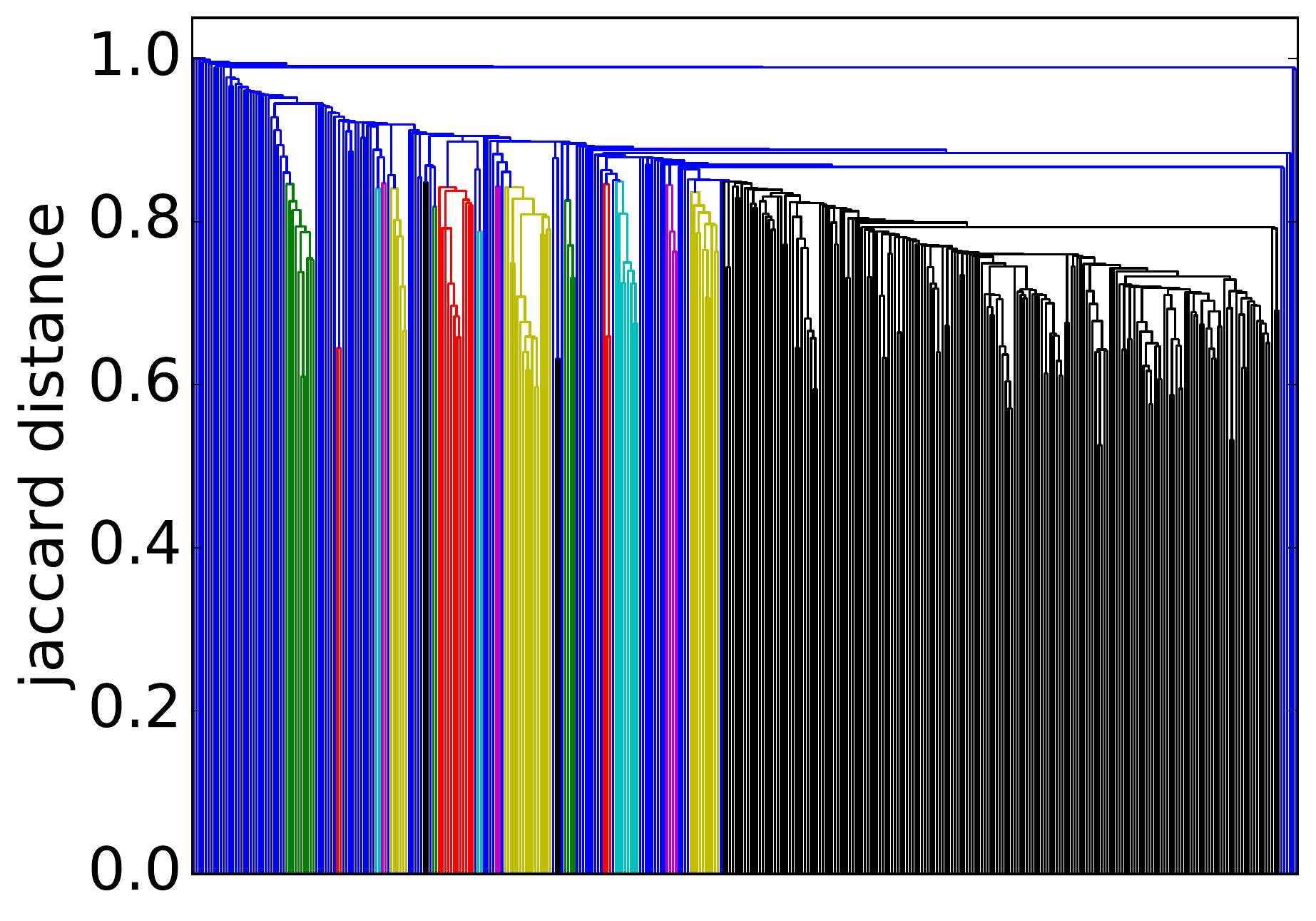}
\caption{Cycling}
\label{fig:cycling_dendro}  
\end{subfigure}
~
% athletics dendro
\begin{subfigure}[t]{0.23\textwidth}
\includegraphics[width=\textwidth]{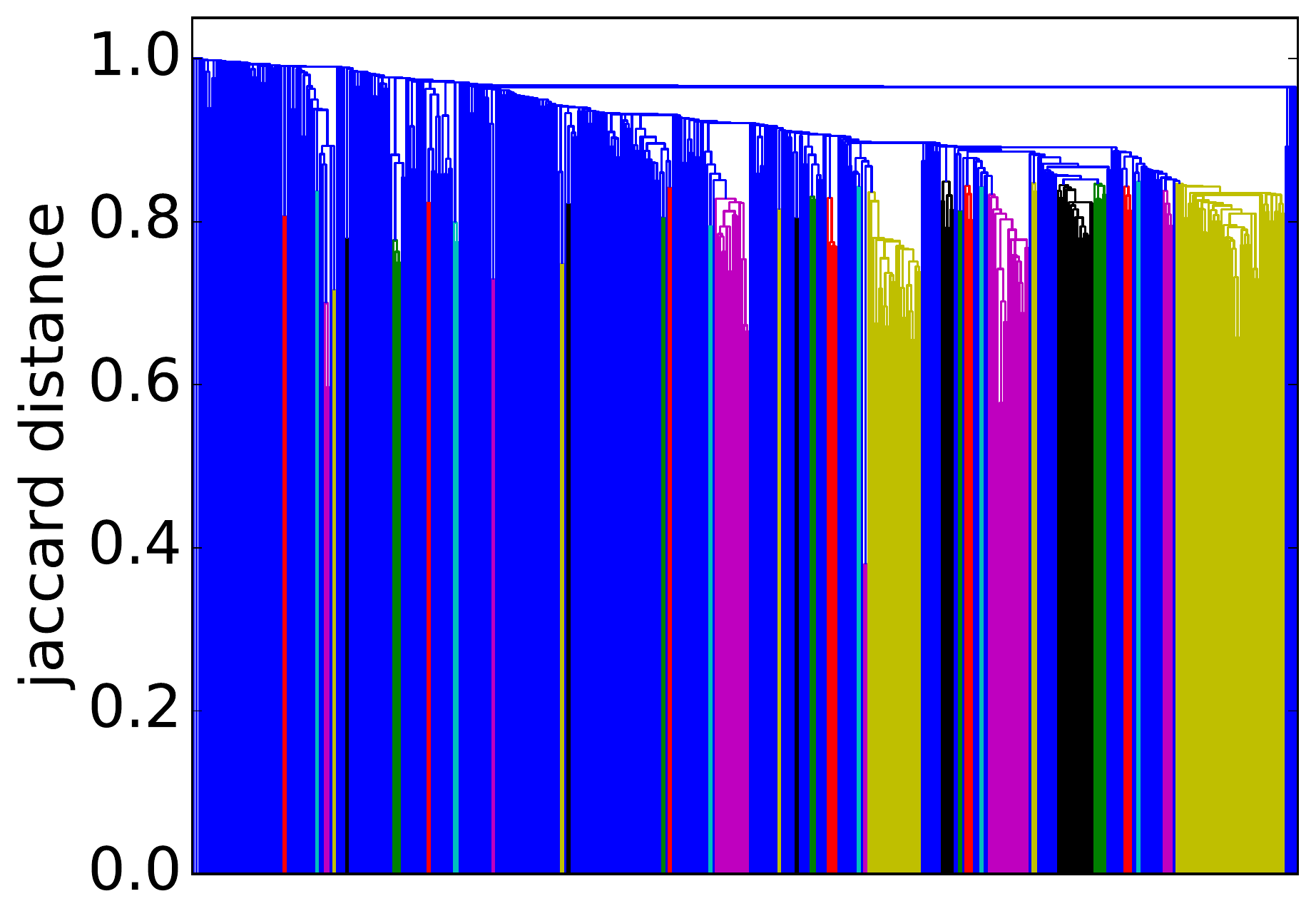}
\caption{Athletics}
\label{fig:athletics_dendro}
\end{subfigure} 
~
% adult actor dendro
\begin{subfigure}[t]{0.23\textwidth}
\includegraphics[width=\textwidth]{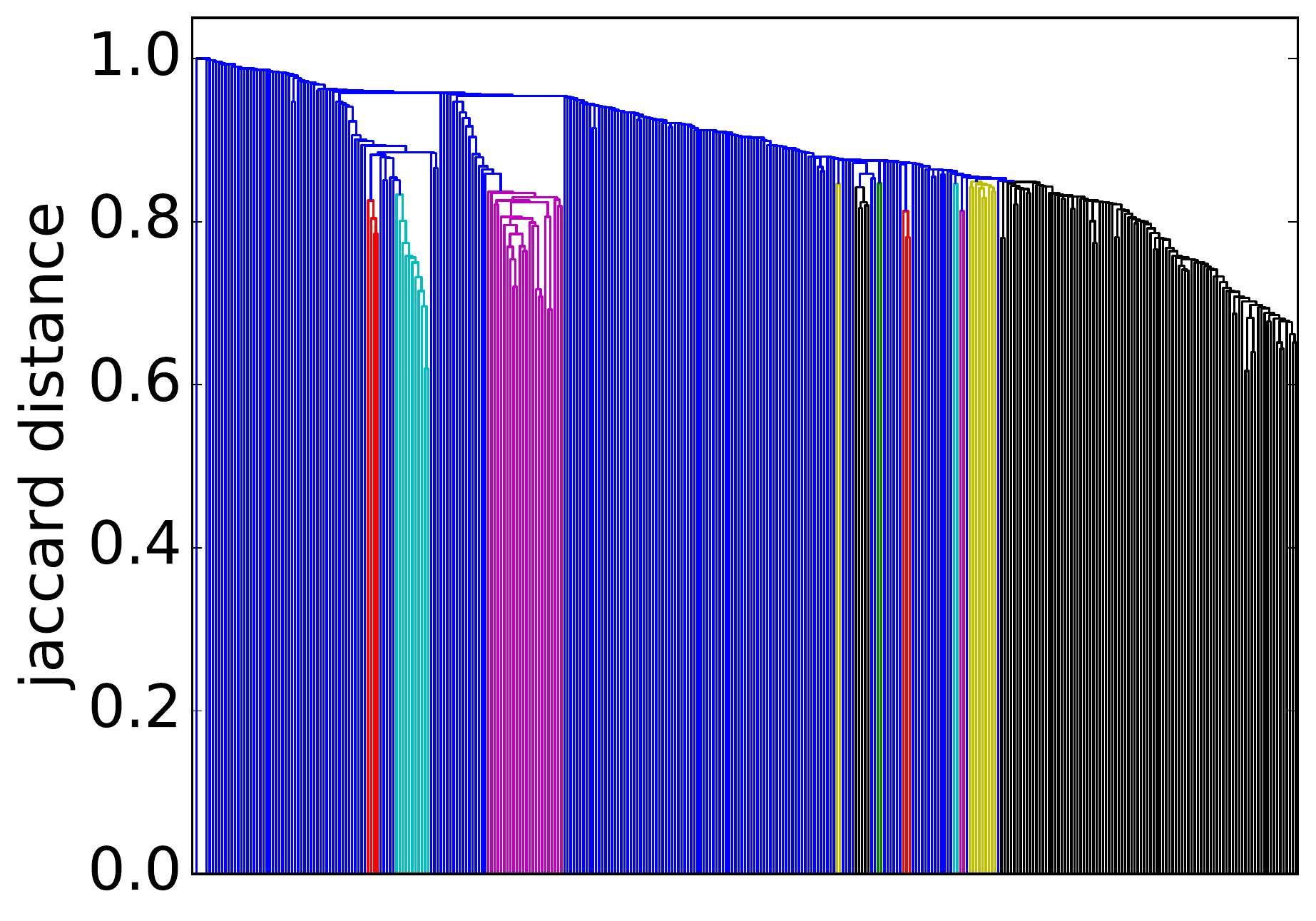}
\caption{Adult Actor}
\label{fig:pornstar_dendro}
\end{subfigure}
\\
Strongly connected communities. Sub-communities mostly due to nationality
\\

\begin{subfigure}[t]{0.23\textwidth}
\includegraphics[width=\hsize,clip]{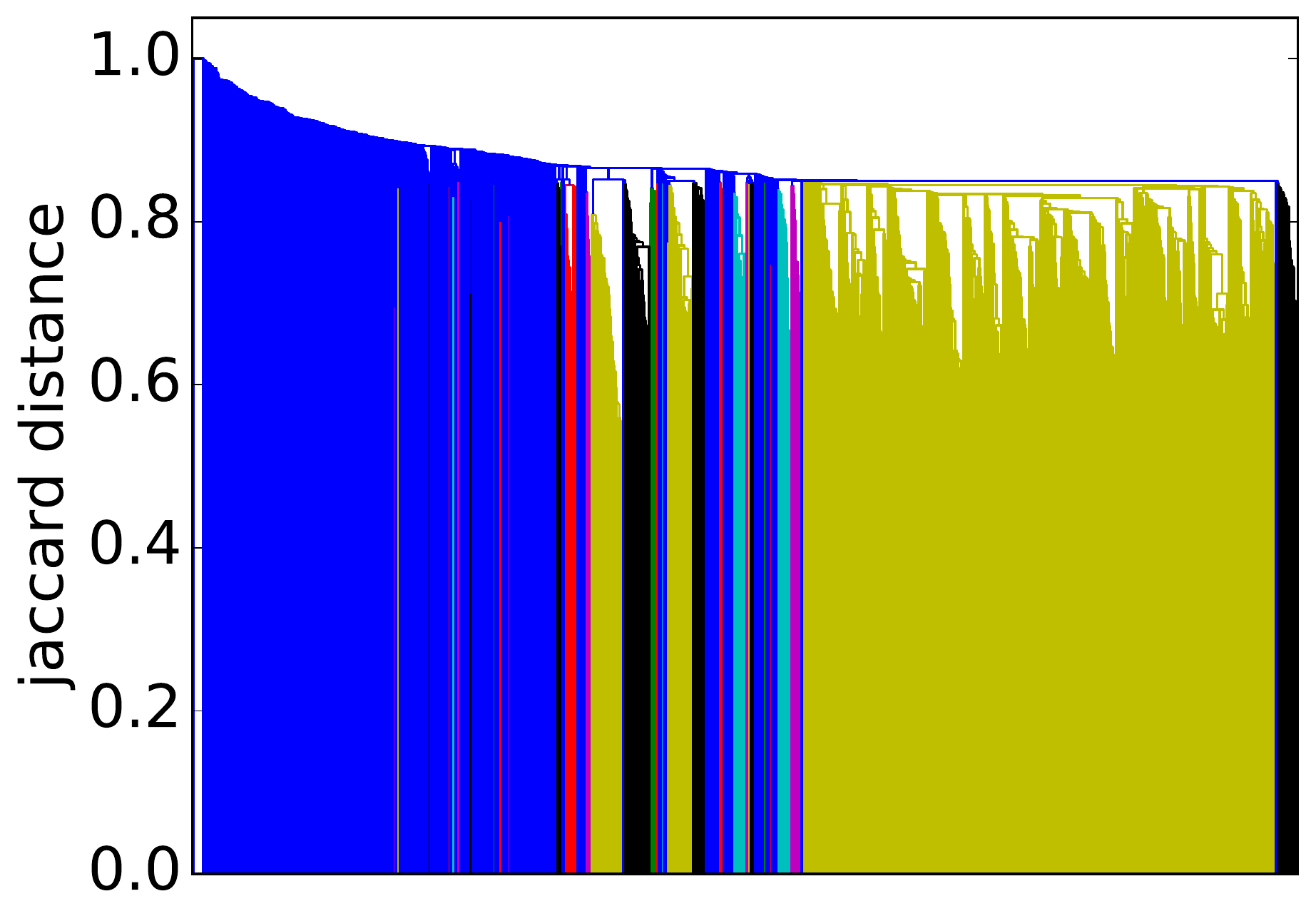}
\caption{American Football}
\label{fig:american_football_dendro}  
\end{subfigure}
~
% baseball dendro
\begin{subfigure}[t]{0.23\textwidth}
\includegraphics[width=\textwidth]{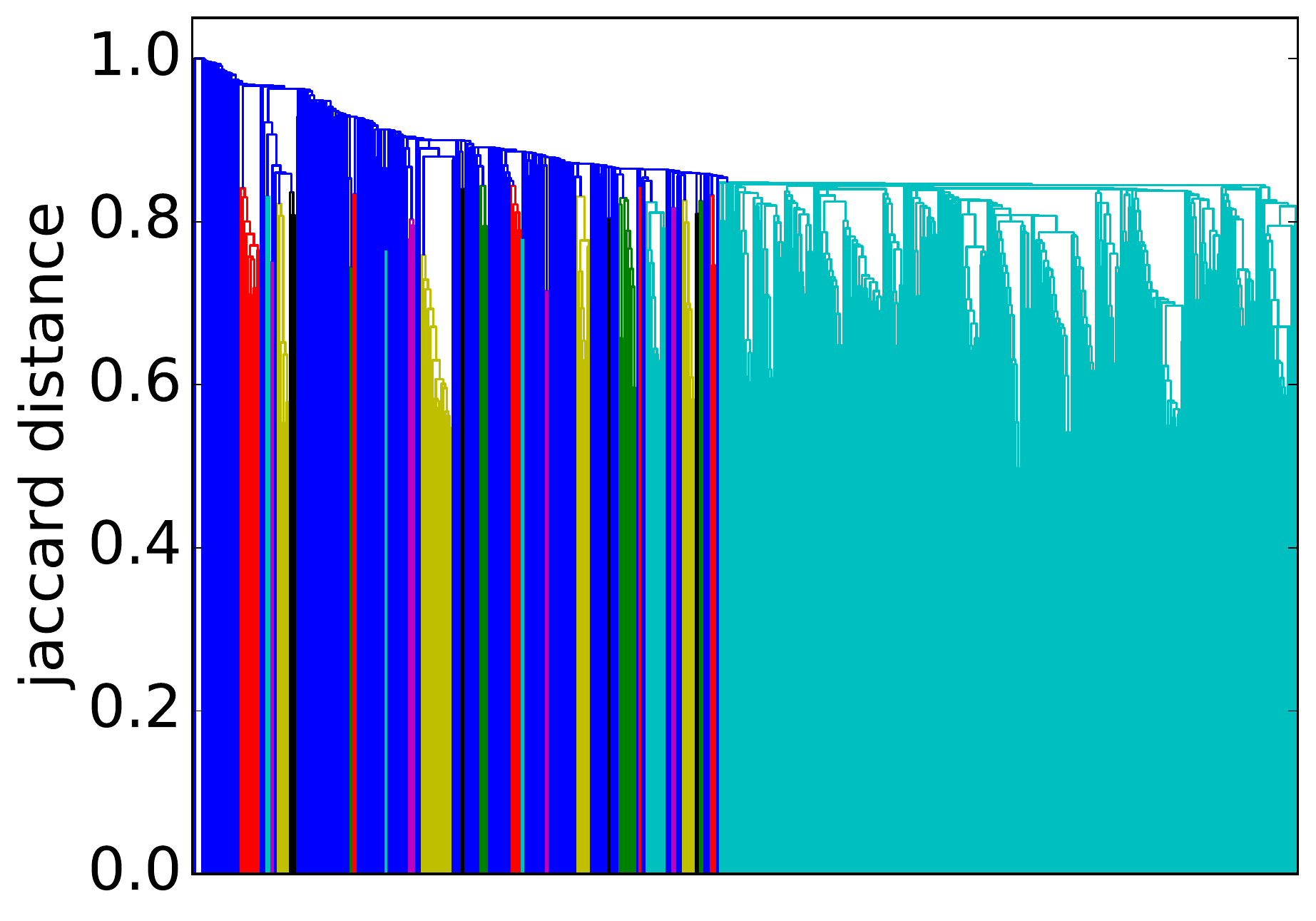}
\caption{Baseball}
\label{fig:baseball_dendro}
\end{subfigure}
~
% basketball dendro
\begin{subfigure}[t]{0.23\textwidth}
\includegraphics[width=\textwidth]{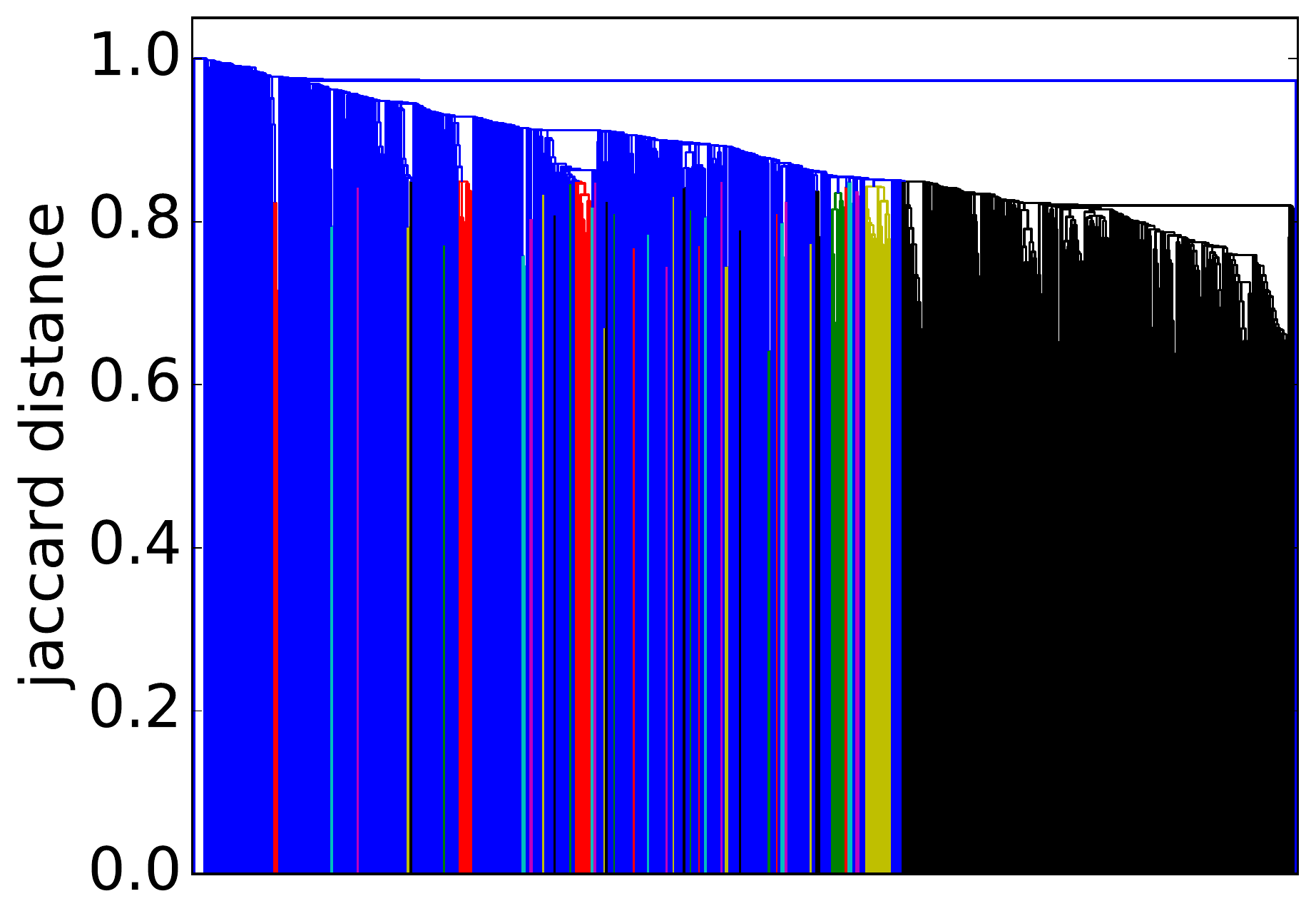}
\caption{Basketball}
\label{fig:basketball_dendro}
\end{subfigure}
~
% football dendro
\begin{subfigure}[t]{0.23\textwidth}
\includegraphics[width=\textwidth]{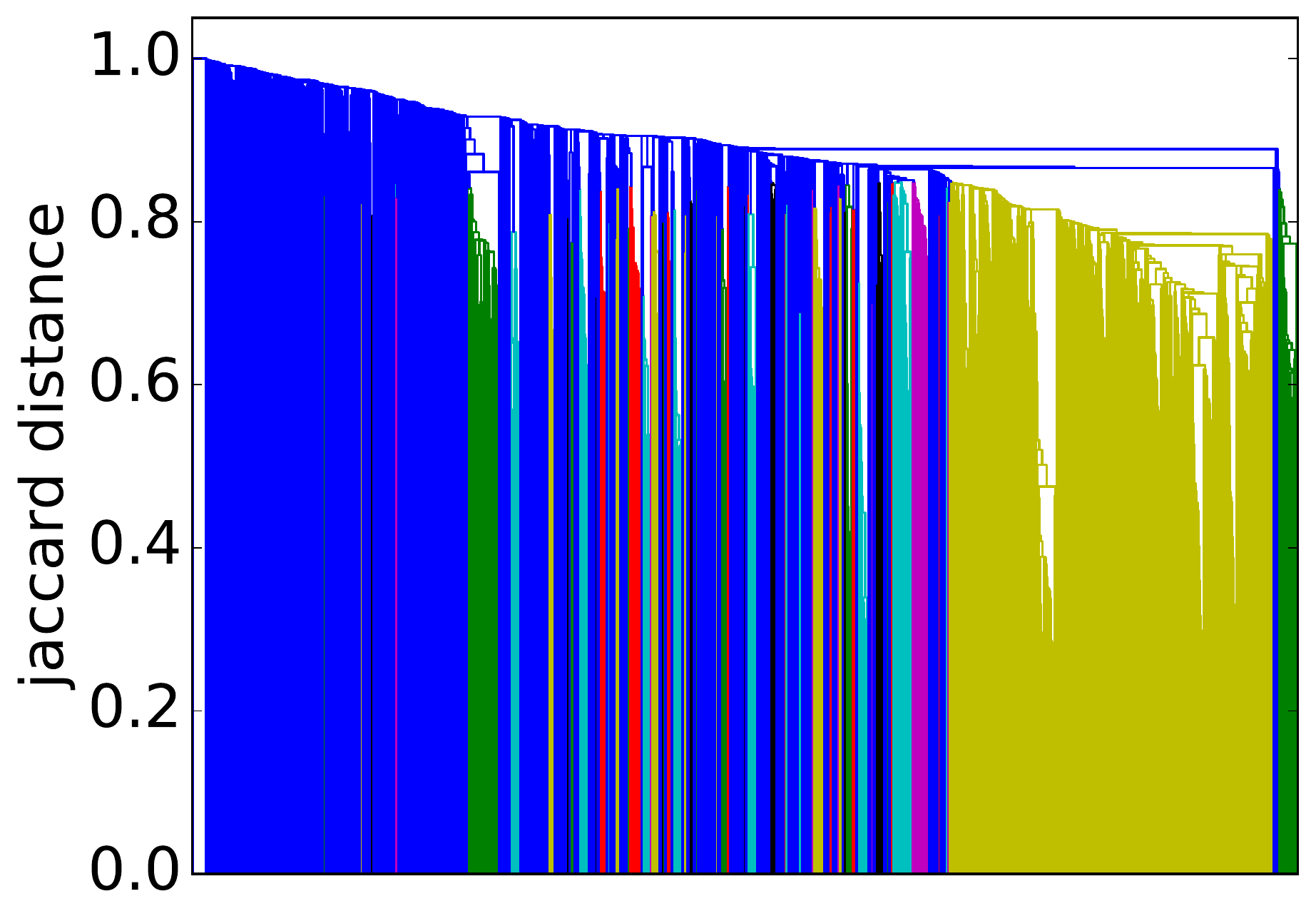}
\caption{Football}
\label{fig:football_dendro}
\end{subfigure} 
\\
Team sports. Many highly connected sub-groups
\\
\end{tabular}
\caption{Dendrograms showing the strength of interconnection within communities. The vertical axes show the Jaccard distances. Blue areas are not strongly connected. In each coloured region no two nodes are separated by a Jaccard distance greater than 0.85. The dendrograms are agglomerative: All accounts with a Jaccard distance less than the $y$-value are fused together into a super-node. The fusing process is sequential and the $x$-axis indicates the order of fusing with the first nodes to agglomerate at the right.}
\label{fig:dendrograms}
\end{figure}
% Gephi diagrams explained
To establish a clearer view of the density and homogeneity of the ground-truth we visualise the communities using network diagrams and dendrograms. Network diagrams are generated in Gephi \citep{Bastian2009a}. The layout uses the Force Atlas 2 algorithm. Colours indicate clusters generated using Gephi's modularity optimisation routine. The node (and label) sizes indicate the weighted degree of each node and are scaled to be between 5 and 20 pixels. The network diagrams reveal any substructure present within the ground-truth. They contain too much information to easily see the individual accounts and so we magnify small subregions and display Twitter profile images for accounts within them. A weakness of the network diagrams is that different edge weights are hard to perceive. To provide a visual representation of the general strength of interaction we generated dendrograms (Figure~\ref{fig:dendrograms}) for each ground-truth community. 
% dendrograms explained
Dendrograms are agglomerative: All accounts with a Jaccard distance less than the $y$-value are fused together into a super-node. Any subgroups containing more than 10 nodes with no two nodes separated by a Jaccard distance greater than 0.85 have been coloured to indicate sub-communities.

\begin{figure}[tb]
\includegraphics[width=\textwidth]{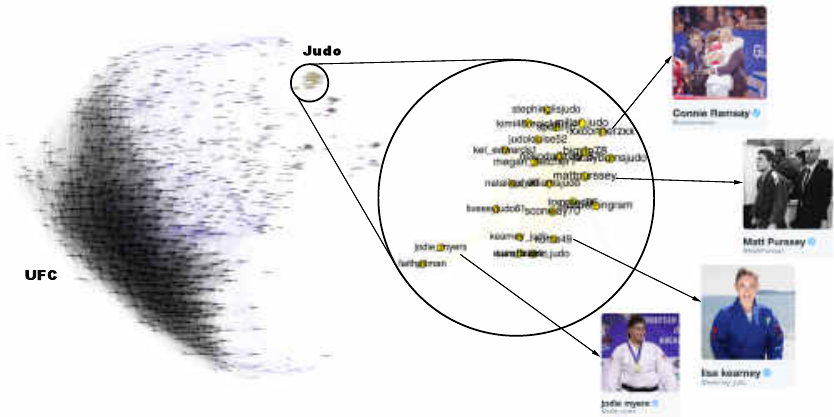}
\caption{The Mixed Martial Arts community is relatively homogeneous and densely interconnected with high clustering and good separability from the rest of the network. The only disconnected region is the yellow region, which has been magnified to show that it is made up of Olympic judo competitors. This community is well detected by all methods.}
\label{fig:taekwondo_network}
\end{figure}
% alcohol / cider
\begin{figure}[tb]
\includegraphics[width=\textwidth]{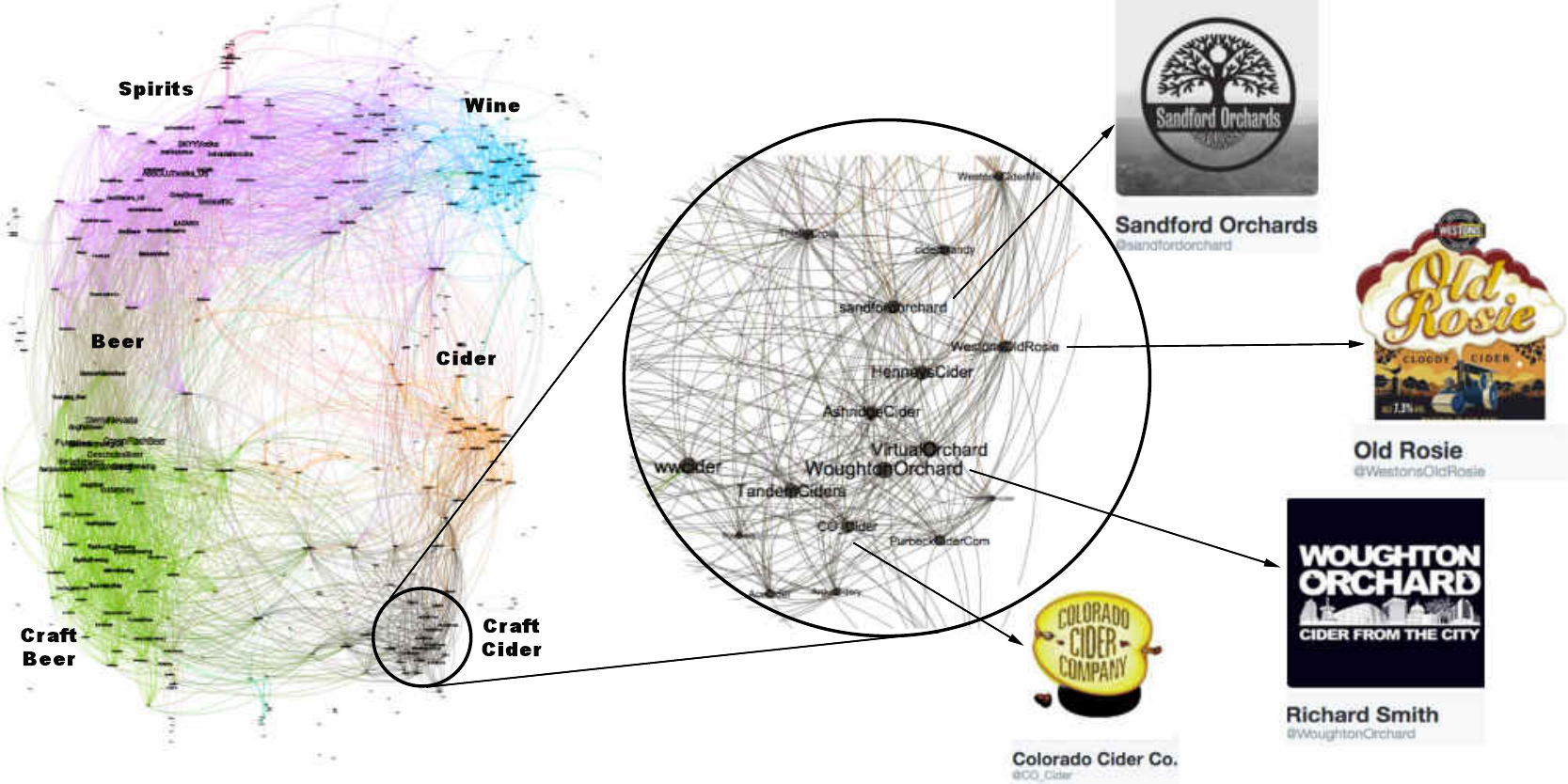}
\caption{The alcohol community is a low density community with poor clustering. It is divided into broad classes of drink such as beer, spirits and wine. We have magnified an area of the cider sub-community}
\label{fig:alcohol_network}
\end{figure}
% hotel / vegas
\begin{figure}[tb]
\includegraphics[width=\textwidth]{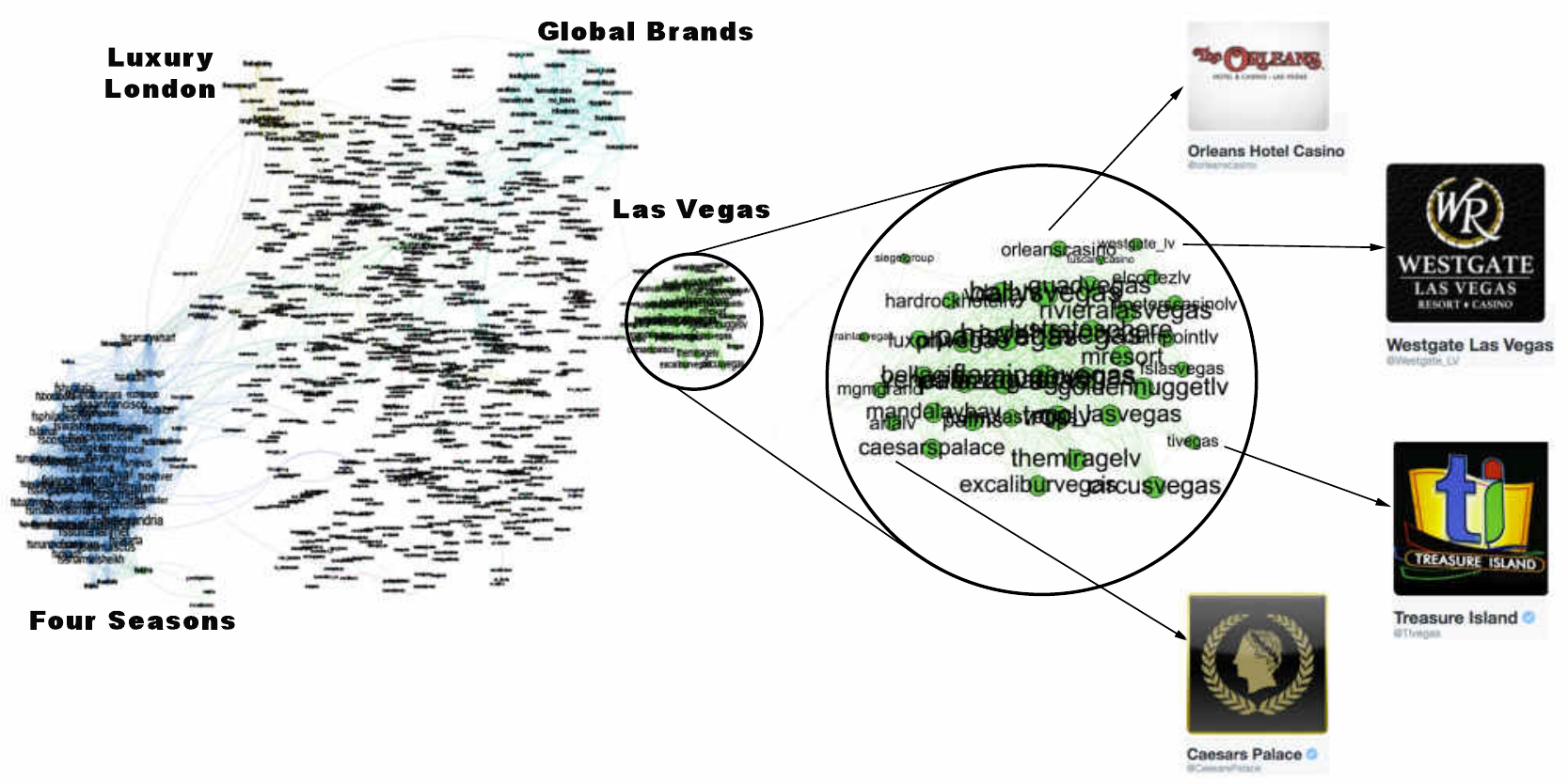}
\caption{The hotels network has low conductance indicating that it is not well separated from the rest of the network. It also has high cohesiveness indicating it contains components that appear to be the true modular units. The two clearly visible subcomponents are the Four Seasons brand in blue to the left and the hotels of Las Vegas, which is magnified}
\label{fig:hotel_network}
\end{figure}
% basketball / WMBA
\begin{figure}[tb]
\includegraphics[width=\textwidth]{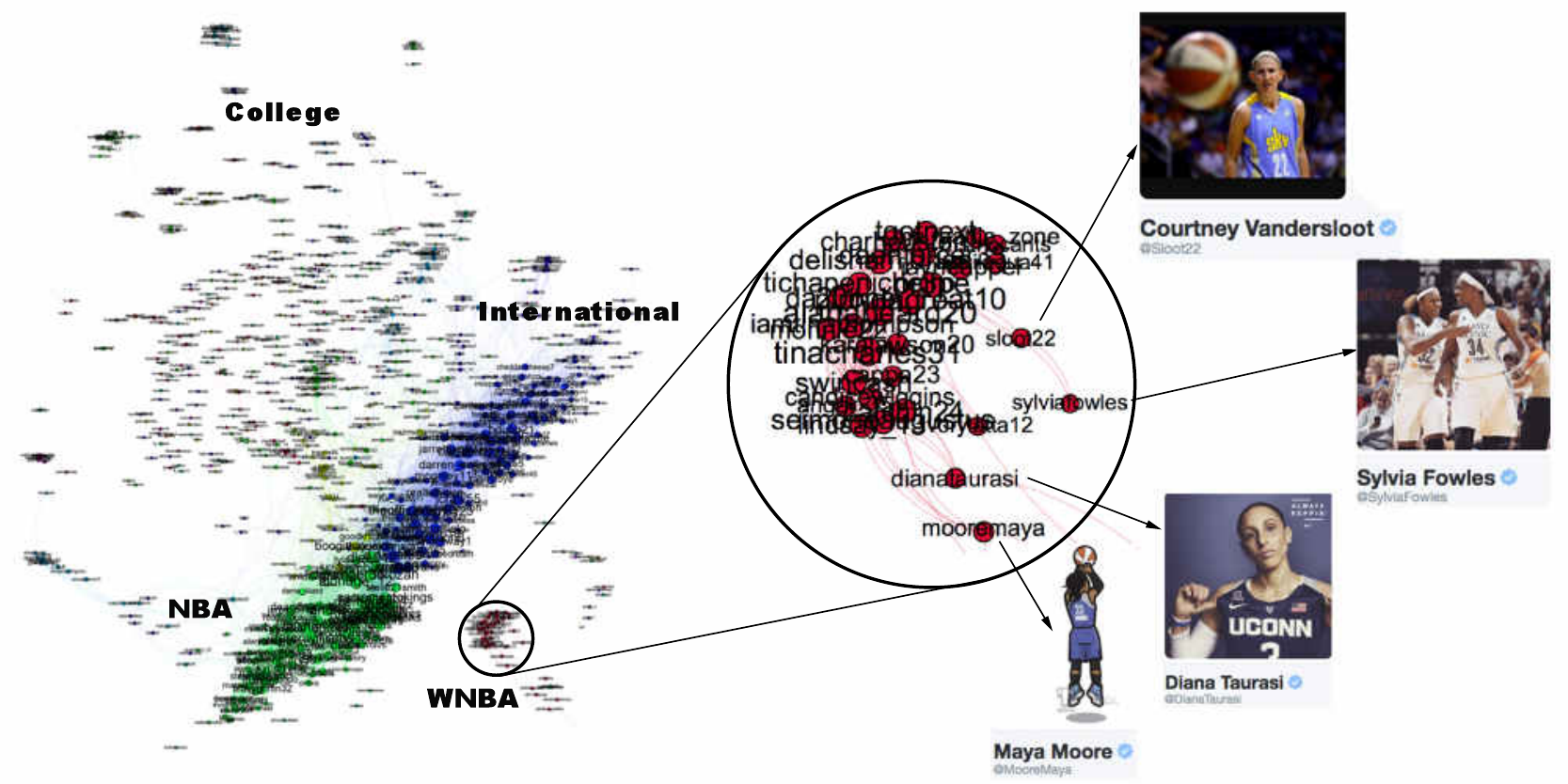}
\caption{The basketball community has very similar attributes to the baseball and american football communities, all being densely connected and well separated from the rest of the network. The individual team structure is not apparent in the graph instead the two large clusters show teams from the eastern and western conferences. The small peripheral clusters are mostly major college teams and we have magnified an area showing players of the Womens National Basketball Association}
\label{fig:basketball_network}
\end{figure}

% qualitative description of ground truth
% best four communities
Figure~\ref{fig:taekwondo_network} shows the Mixed Martial Arts (MMA) community. From Table~\ref{tab:communities} we see that this community is densely connected, strongly clustered and very well separated from the rest of the network. The black region in Dendrogram~\ref{fig:taekwondo_dendro} is a massive cluster where the distance between any two nodes is less than 0.8. It depicts MMA fighters, mostly fighting in the Ultimate Fighting Championship (UFC). There is a single well separated sub-community, which is magnified in Figure~\ref{fig:taekwondo_network} showing Olympic judo fighters. MMA is the \emph{best} community in our study. The Cycling, Adult Actor and Athletics communities are similar in structure to MMA (See Table~\ref{tab:communities}).

% team sport communities
Figure~\ref{fig:basketball_network} shows that the basketball community (largely NBA players) exhibits two large communities (the two NBA conferences). The individual team structure within the divisions is apparent from the fine banding in Figure~\ref{fig:basketball_dendro} where many well-connected sub-clusters, each with a distance of less than 0.85 between all pairs of nodes, are visible. We have magnified a small disconnected region of Figure~\ref{fig:basketball_network}, which shows players of the Women's National Basketball Association (WNBA). Baseball, football and american football exhibit similar structural properties. 

% generally bad communities
Figure~\ref{fig:alcohol_network} shows that the industry is split into four major groups representing the different classes of alcoholic drink (wine, beer, cider and spirits). We have magnified a region of the network that contains mostly English \emph{craft} ciders. Dendrogram~\ref{fig:alcohol_dendro} shows that the alcohol network is mostly poorly connected with only two coloured regions indicating well connected sub-communities. From Table~\ref{tab:communities} it can be seen that the alcohol network exhibits a low link density and separability, indicating that the community lacks distinction from the rest of the network. This is a consistent pattern for communities drawn from industrial segmentations. 

Figure~\ref{fig:hotel_network} shows an example of the final group of ground truth communities: industrial groups with prominent sub-communities. In this case the major sub-communities are the Four Seasons Hotel group and hotels located in Las Vegas (magnified). Dendrogram~\ref{fig:hotel_dendro} shows that while the hotel network is generally poorly connected, there are sizeable highly interconnected sub-communities. From Table~\ref{tab:communities} shows that the hotel community has low clustering as most accounts are disconnected, high cohesiveness as there are well connected sub-groups and a low Conductance Ratio. The travel, airlines and cosmetics communities all share these traits.

% Figure \ref{fig:athletics_network} shows that the athletics community is made up of many disconnected nodes and five well defined sub-commu\-nities. Unlike MMA and American football, which are largely sports played in the USA, athletics is global, and the principal clusters represent athletes from the USA, Australia and the UK\footnote{Our labelled data has an Anglo-Saxon bias}. Furthermore, the discipline of athletics is really many separate sports, and triathlon forms another well separated cluster. In this case, Figure \ref{fig:results} shows that AC performs best as there are really five sub-clusters, and with 30 seeds there is a high probability of seed representation in each one.

In summary we identify four groups of ground-truth communities and evaluate their quality based on the four axioms. We find that the groups differ greatly in quality. The group containing mixed martial arts, cycling, athletics and adult actors satisfies the four axioms and form a good set of ground-truth for algorithm evaluation. The group comprising team sports (american football, baseball, basketball and football) satisfy three of the four axioms (they are not homogenous). The remaining communities only contain sub-groups that satisfy any of the axioms.

\section{Experimental Evaluation}
\label{sec:evaluation}

% intro - results for the minhash approximation and the community detection piece
Our approach to real-time community detection relies on two approximations: minhashing for rapid Jaccard estimation and locality sensitive hashing to provide a fast query mechanism on top of minhashing. We assess the effect of these approximations, and demonstrate the quality of our results in three experiments: 
(1) We measure the sensitivity of the Jaccard similarity estimates with respect to the number of hash functions used to generate the signatures. This will justify the use of the minhash approximation for computing approximate Jaccard similarities. 
(2) We compare the run time and recall of our process on ground-truth communities against the Personal Page Rank (PPR) algorithm (state of the art) on a single laptop.
(3) We visualise detected communities and demonstrate that association maps for social networks using minhashing and LSH produce intuitively interpretable maps of the Twitter and Facebook graphs in real-time on a single machine.

\subsection{Experiment 1: Assessing the Quality of Jaccard Estimates}
\label{sec:eval_exp1}
% Minhash results
\begin{table}[tb]
  \centering
  \caption{The Twitter accounts with the highest Jaccard similarities to @Nike. $J$ and $R$ give the true Jaccard coefficient and Rank, respectively. $\hat{J}$ and $\hat{R}$ give approximations using Equation \eqref{eq:min_est} where the superscript determines the number of hashes used. Signatures of length 1,000 largely recover the true Rank.}
\scalebox{1}{
    \begin{tabular}{rrrrrrr}
    \toprule
    Twitter handle & $J$     & $R$     & $\hat J^{100}$  & $\hat R^{100}$  & $\hat J^{1000}$ & $\hat R^{1000}$ \\
    \midrule
    adidas & 0.261 & 1     & 0.22  & 2     & 0.265 & 1 \\
    nikestore & 0.246 & 2     & 0.25  & 1     & 0.255 & 2 \\
    adidasoriginals & 0.200 & 3     & 0.18  & 3     & 0.222 & 3 \\
    Jumpman23 & 0.172 & 4     & 0.13  & 7     & 0.166 & 4 \\
    nikesportswear & 0.147 & 5    & 0.18  & 4     & 0.137 & 5 \\
    nikebasketball & 0.144 & 6     & 0.16  & 5     & 0.127 & 7 \\
    PUMA  & 0.132 & 7     & 0.13  & 6     & 0.132 & 6 \\
    nikefootball & 0.127 & 8     & 0.08  & 17    & 0.110  & 9 \\
    adidasfootball & 0.112 & 9     & 0.09  & 16    & 0.113 & 8 \\
    footlocker & 0.096 & 10    & 0.08  & 17    & 0.096 & 11 \\
    \bottomrule
    \end{tabular}%
    }
  \label{tab:Nike}%
\end{table}%

\begin{figure}[tb]
  \centering
    \includegraphics[width=0.55\hsize]{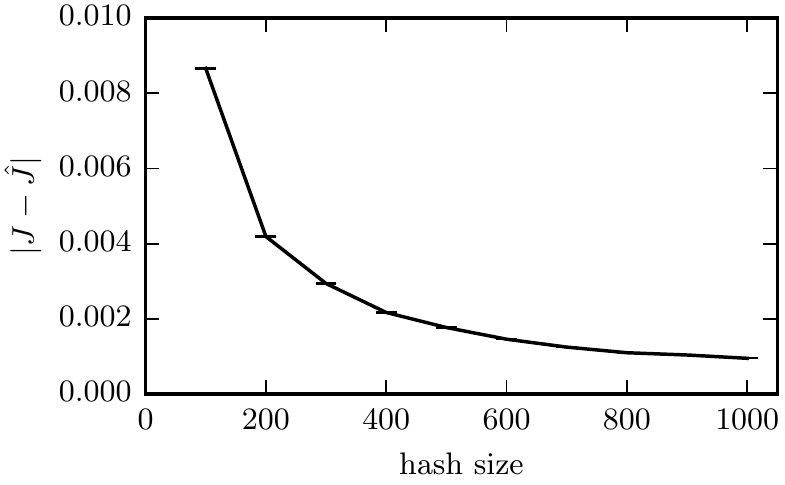}
    \caption{Expected error from Jaccard estimation using minhash signatures against the number of the hashes used in the signature. The error bars show twice the standard error using 400,000 data points.}
\label{fig:minhash}
%\figspace
\end{figure}

We empirically evaluate the minhash estimation error using a sample of 400,000 similarities taken from the 250 billion pairwise relationships between the Twitter accounts in our study. We compare estimates using Equation~\ref{eq:min_est} to exact Jaccards obtained by exhaustive calculations on the full sets using Equation~\ref{eq:jaccard}. Figure~\ref{fig:minhash} shows the estimation error (L1 norm) as a function of the number of hashes comprising the minhash signature. Standard error bars are just visible up until 400 hashes. The graph shows an expected error in the Jaccard of just 0.001 at 1,000 hashes. The high degree of accuracy and diminishing improvements at this point led us to select a signature length of $K = 1,000$. This value provides an appropriate balance between accuracy and performance (both runtime and memory scale linearly with $K$).

% The things similar to Nike look correct
A top-ten list of Jaccard similarities is given in Table~\ref{tab:Nike} for the Nike Twitter account (based on the true Jaccard). Possible matches include sports people, musicians, actors, politicians, educational institutions, media platforms and businesses from all sectors of the economy. Of these, our approach identified four of Nike's biggest competitors, five Nike sub-brands and a major retailer of Nike products as the most associated accounts. This is consistent with our assertion that the Jaccard similarity of neighbourhood sets provides a robust similarity measure between accounts. We found similar trends throughout the data and this is consistent with the experience of analysts at Starcount, a London based social media analytics company, who are using the tool. Table~\ref{tab:Nike} also shows how the size of the minhash signature affects the Jaccard estimate and the corresponding rank of similar accounts. Local community detection algorithms  add accounts in similarity order. Therefore, approximating the true ordering is an important property. We measure the Spearman rank correlation between the true Jaccard similarities (column $R$) and those calculated from signatures of length 100 (column $\hat R^{100}$) and 1000 (column $\hat R^{1000}$) to be 0.89 and 0.97 respectively. The close correspondence of the rank vector using signatures of length 1,000  and the true rank  supports our decision to use signatures of containing 1,000 hashes.

\subsection{Experiment 2: Comparison of Community Detection with PPR}

In experiment 2 we move from assessing a single component (minhashing) to a system-wide evaluation: We evaluate the ability of our algorithm to detect related entities by measuring its performance as a local community detection algorithm seeded with members of the ground-truth communities listed in Table~\ref{tab:communities}. As a baseline for comparison we use the PPR algorithm, which is considered to be the state of the art for this problem \citep{Kloumann2014}. It is impossible to provide a fully like-for-like comparison with PPR: Running PPR on the full graph (700 million vertices and 20 billion edges) that we extract features from requires cluster computing and could return results outside of the accounts we considered. The alternative is to restrict PPR to run on the directly observed network of the 675,000 largest Twitter accounts, which could then be run on a single machine. We adopt this latter approach as it is the only option that meets our requirements (single machine and real-time).

\begin{figure}[tb]
\centering
\includegraphics[width=0.7\textwidth]{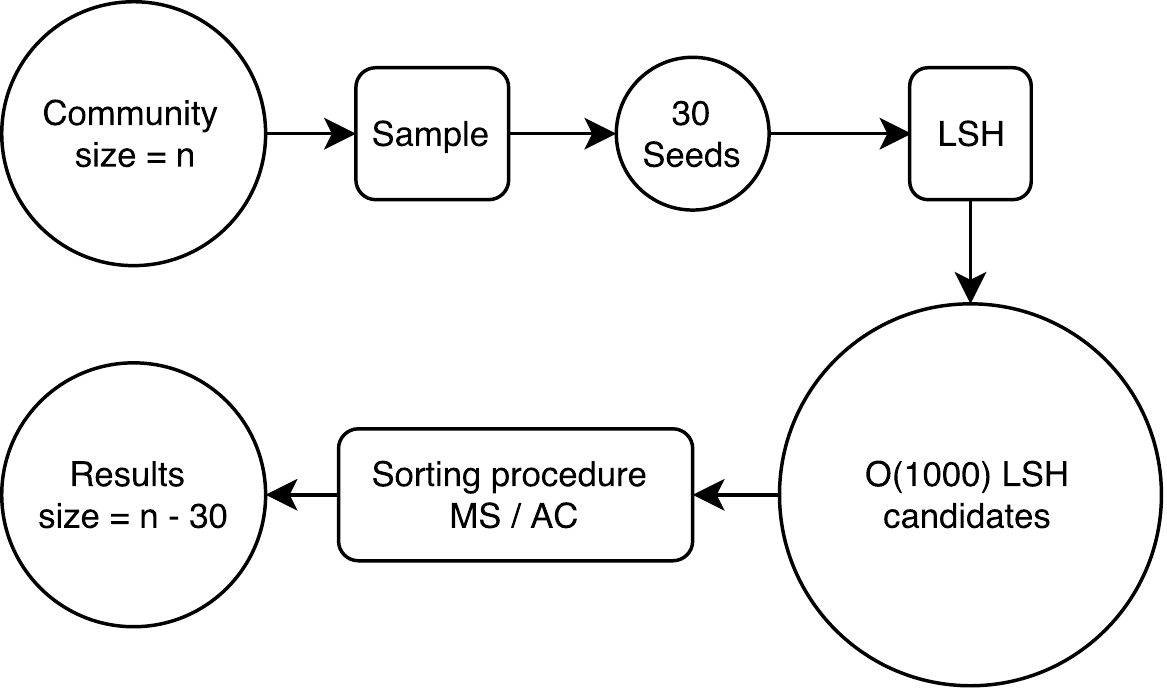}
\caption{Process diagram for Experiment 2. Circles indicate (intermediate) results and rectangles are processes. The size of circles is illustrative of the number of accounts at each stage. From each community we sample 30 seeds to use as input for an LSH query. The LSH query produces a larger set of candidate accounts. The candidate lists are submitted to the MS and AC sorting procedures, which return results.}
\label{fig:exp2}
\end{figure}
%\todo[inline]{Although Fig. 10 is now much better it is still not visually appealing. You can delete this note when you read it. Your decision how important you think a visually appealing representation is.}
% explain experimental goal for community expansion
In our experimentation, we randomly sampled 30 seeds from each ground-truth community. To produce MS and AC results we followed the process depicted in Figure~\ref{fig:exp2}: The seeds are input to an LSH query, which produces a list of candidate near-neighbours. For each candidate the Jaccard similarity is estimated using minhash signatures and sorted by either the MS or AC procedures.  

% our PPR implementation
We compare MS and AC to PPR operating on the directly observed network of the 675,000 largest accounts. Our PPR implementation uses the 30 seeds as the \emph{teleport} set and runs for three iterations returning a ranked list of similar Twitter accounts. 

In all cases, we sequentially select accounts in similarity order and measure the recall after each selection. 
The recall is given by
\begin{align}
\text{recall} = \frac{\left|C\cap C_{\text{true}} \right |}{| C_{\text{true}}| -  |C_{\text{init}}|} 
\end{align}
with $C_{\text{init}}$ as the initial seed set, $C_{\text{true}}$ as the ground truth community and $C$ as the set of accounts added to the output. For a community of size $|C|$ we do this for the $|C|-30$ most similar accounts so that a perfect system could achieve a recall of one.

% Graphs
\begin{figure}[tb]
  \centering
    \includegraphics[width=\textwidth]{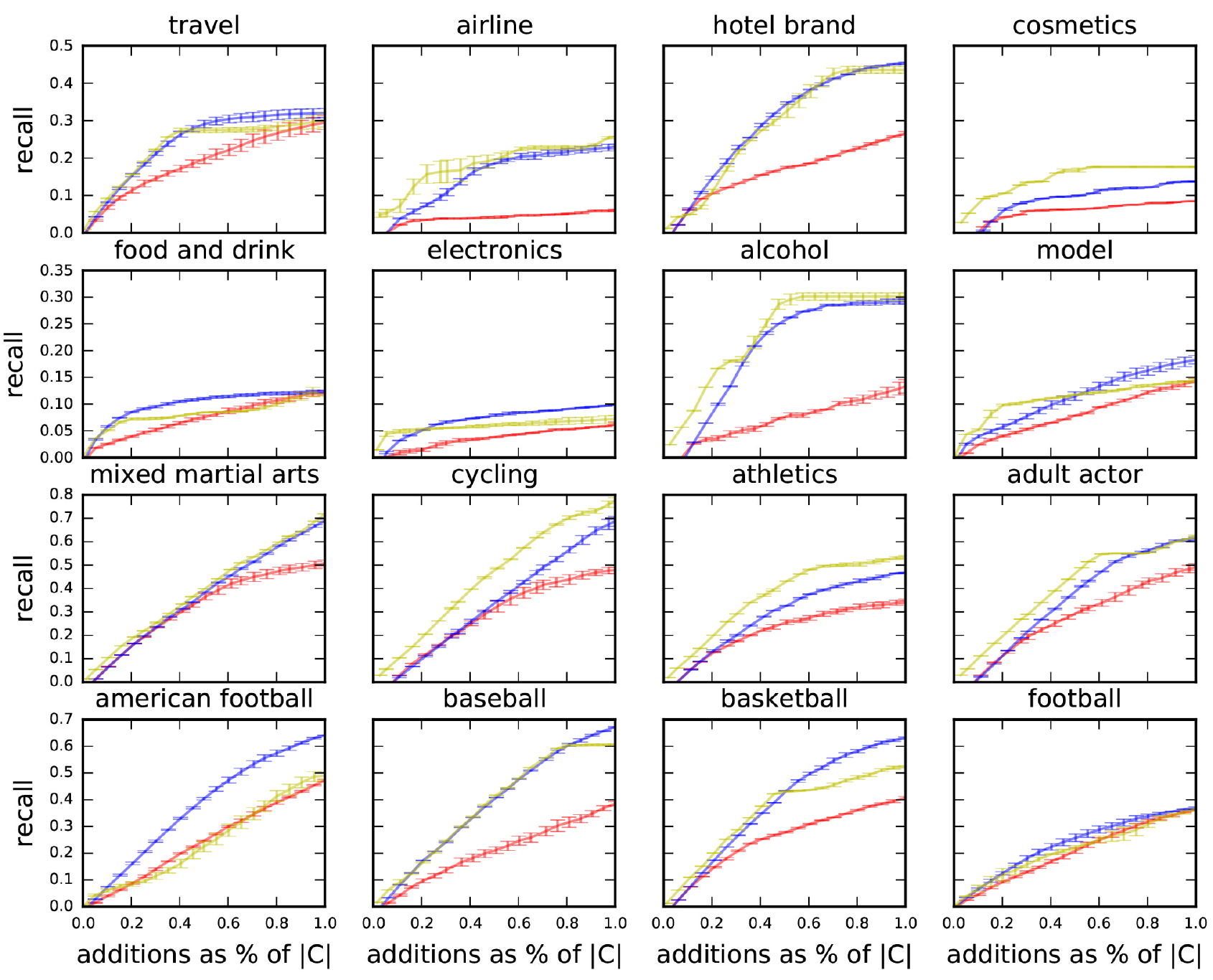}
\caption{Average recall (with standard errors) of Agglomerative Clustering (yellow), Personal PageRank (red) and Minhash Similarity (blue) against the number of additions to the community expressed as a fraction of the size of the ground-truth communities given in Table~\ref{tab:communities}. The tight error bars indicate that the methods are robust to the choice of seeds.}
\label{fig:results}
\end{figure}

\begin{table}[tb]
\centering
  \caption{Area under the recall curves (Figure~\ref{fig:results}). Bold entries indicate the best performing method. Minhash similarity (MS) is the best method in 8 cases, Agglomerative Clustering (AC) in 8 cases and  Personalised PageRank (PPR) in none. A perfect community detector would score 0.5}
\label{tab:AUC}%
\begin{tabular}{l|lll}
\toprule
\textbf{Tags}               & \textbf{PPR} & \textbf{MS} & \textbf{AC} \\
\midrule
\textbf{travel}             & 0.186        & \textbf{0.240 }      & 0.230       \\
\textbf{airline}            & 0.040        & 0.151       & \textbf{0.180 }      \\
\textbf{hotel brand}        & 0.160        & \textbf{0.294}       & 0.285       \\
\textbf{cosmetics}          & 0.055        & 0.086       & \textbf{0.143 }      \\
\textbf{food and drink}     & 0.072        & \textbf{0.099 }      & 0.082       \\
\textbf{electronics}        & 0.035        & \textbf{0.069}       & 0.059       \\
\textbf{alcohol}            & 0.069        & 0.199       & \textbf{0.229}       \\
\textbf{model}              & 0.078        & \textbf{0.110}       & 0.109       \\
\textbf{mixed martial arts} & 0.317        & 0.363       & \textbf{0.386}       \\
\textbf{cycling}            & 0.278        & 0.330       & \textbf{0.445}       \\
\textbf{athletics}          & 0.219        & 0.285       & \textbf{0.365}       \\
\textbf{adult actor}        & 0.269        & 0.347       & \textbf{0.397}       \\
\textbf{american football}  & 0.240        & \textbf{0.371}       & 0.240       \\
\textbf{baseball}           & 0.203        & \textbf{0.379}       & 0.378       \\
\textbf{basketball}         & 0.252        & \textbf{0.380}       & 0.353       \\
\textbf{football}           & 0.202        & \textbf{0.233}       & 0.212  \\ 
\bottomrule
\end{tabular}
\end{table}

% description of results 
%\todo[inline]{This paragraph is new, can you check it makes sense - Is there enough description of the charts and AUC table now?}
The results of this experiment are shown in Figure~\ref{fig:results} with the Area Under the Curves (AUC) given in Table~\ref{tab:AUC}.  Bold entries in Table~\ref{tab:AUC} indicate the best performing method. In all cases MS and AC give superior results to PPR. 

Figure~\ref{fig:results} shows standard errors over five randomly chosen input sets of 30 accounts from $C_{true}$. The confidence bounds are tight indicating that the methods are robust to the choice of input seeds. Figure~\ref{fig:results} is grouped to correspond to the dendrograms in Figure~\ref{fig:dendrograms}. Performance of all methods is considerable affected by the \emph{quality} of the communities. Communities with good values of the metrics given in Table~\ref{tab:communities} in general have superior recall across all methods. The third row of Figure~\ref{fig:results} contains the \emph{best} communities as measured by the metrics in Table~\ref{tab:communities}. For this group recalls are as high as 80\% (Cycling, AC). The worst group of communities are the transnational industrial communities in the second row. The lowest recall in row three (Athletics PPR) is still higher than the highest recall in the second row of results (Alcohol, AC). The best performing method for every community in row three of the results is AC. This is because AC is an adaptive method that can incorporate information from early results. The downside of an adaptive method is that pollution from false positives can rapidly degrade performance. This can be seen in the step decrease in gradient of the AC curves for basketball, baseball and adult actors. The fourth row of the table contains team sports. Team sports also have good metrics in Table~\ref{tab:communities}, but differ markedly in structure from the communities in row three. The communities in row four have well defined multi-modal sub-structures generated by the different teams. Both AC and MS are unimodal procedures that store the centre of a set of data points. For a multimodal distribution the mean may not be particularly close to the distribution and so false positives will occur. As AC incorporates false positives into the estimation procedure for all future results MS outperforms AC for all team sport communities. Of the communities in the first and second rows of Figure~\ref{fig:results} AC is best performing in four and MS is best performing in four. These communities are all diffuse, but some have a single densely connected region that can be found well by AC.  

% first / second order connections etc
Much of the difference in performance of these methods derives from their respective ability to explore the graph: PPR is really a global algorithm that has been modified to find local relationships. After three iterations PPR uses both first, second and third-order connections. First-order connection methods just use edges that directly connect to the seed nodes (neighbours). Second-order methods also give weight to the connections of the first-order nodes (neighbours of neighbours) and so on for third-order connections. The ability to explore higher-order connections is the principal reason identified by \cite{Kloumann2014} for the state-of-the-art performance of PPR. They also note that after two iterations most of the benefit is realised and that after three iterations there is no more improvement.  Our implementations of MS and AC are effectively second-order methods as they operate on a derived graph where the edge weight between two vertices is calculated from the overlap of the respective neighbourhoods. MS and AC outperform PPR because they are based on many more second order connections as they run on a compressed version of the full graph instead of a sub-graph.
PPR is expected to perform better given more computational resources, but the additional complexity, run time, latency or financial cost required for any scaled up/out solution would violate our system constraints.

% runtime analysis
Table~\ref{tab:runtimes} gives the mean and standard deviation of the run times averaged over the 16 communities.
\begin{table}[tb]
  \centering
  \caption{Clustering runtimes averaged over communities.}
    \begin{tabular}{rrr}
    \toprule
    \multicolumn{1}{c}{Method} & \multicolumn{1}{c}{Mean(s) } & \multicolumn{1}{c}{Std.Dev.} \\
    \midrule
    \multicolumn{1}{c}{PPR} & \multicolumn{1}{c}{12.58} & \multicolumn{1}{c}{8.83} \\
    \multicolumn{1}{c}{MS} & \multicolumn{1}{c}{0.23} & \multicolumn{1}{c}{0.08} \\
    \multicolumn{1}{c}{AC} & \multicolumn{1}{c}{18.6} & \multicolumn{1}{c}{22.0} \\
    \bottomrule
    \end{tabular}%
  \label{tab:runtimes}%
\end{table}%
MS is the fastest method by two orders of magnitude. Average human reaction times are approximately a quarter of a second and so MS delivers a real-time user experience \citep{Hewett1992}. As MS is the only method capable of operating in the real-time domain and this is a system requirement, we choose the MS procedure for experiment 3 and in our operational prototype.

\subsection{Experiment 3: Real-Time Graph Analysis and Visualisation}

\begin{figure}[tb]
\centering
\includegraphics[width = 0.7\hsize]{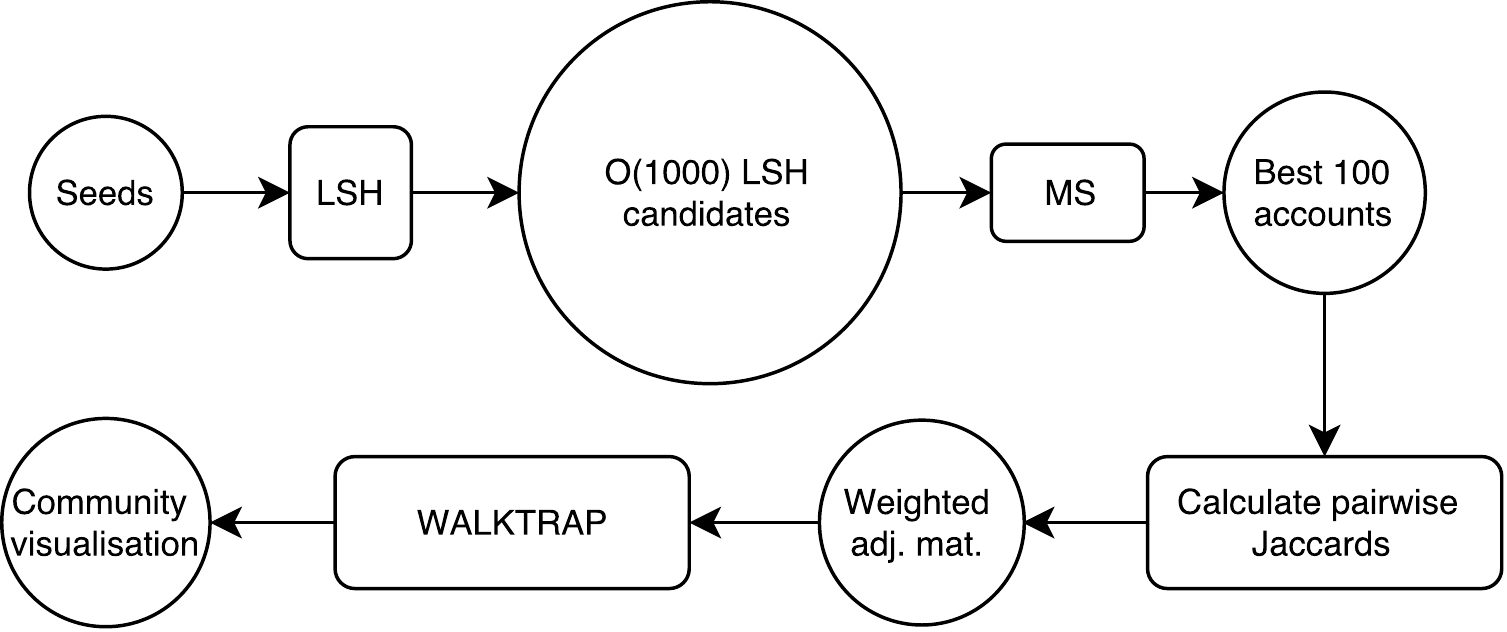}
\caption{Process diagram for Experiment 3. A set of seeds is queried using LSH and Minhash Similarity. The weighted adjacency matrix for the top 100 results is estimated using minhash signatures. The WALKTRAP community detection algorithm is run on the weighted adjacency matrix and the results are visualised.}
\label{fig:exp3}
\end{figure}

In the following, we provide example applications of our system to graph analysis. Users need only input a set of seeds, wait a quarter of a second and the system discovers the structure of the graph in the region of the seeds. Users can then iterate the input seeds based on what previous outputs reveal about the graph structure. Figure~\ref{fig:fb_graphs} shows results on the Facebook Page Engagements network while Figures~\ref{fig:twit_graphs1} and~\ref{fig:twit_graphs} use the Twitter Followers graph. 
% diagram generation process
Each diagram is generated by the procedure shown in Figure~\ref{fig:exp3}: Seeds are passed to the MS process, which returns the 100 most related entities. All pairwise Jaccard estimates are then calculated using the minhash signatures and the resulting weighted adjacency matrix is passed to the WALKTRAP global community detection algorithm. The result is a weighted graph with community affiliations for each vertex. In our visualisations we use the Force Atlas 2 algorithm to lay out the vertices. The thickness of the edges between vertices represents the pairwise Jaccard similarity, which has been thresholded for image clarity. The vertex size represents the weighted degree of the vertex, but is logarithmically scaled to be between 1 and 50 pixels. The vertex colours depict the different communities found by the WALKTRAP community detection algorithm.

% comparison facebook/twitter
We show some results using the Facebook Pages engagement graph to demonstrate that our work is broadly applicable across digital social networks. However there are some key differences between the Facebook Pages engagement graph and the Twitter Followers graph. As Following is the method used to subscribe to a Twitter feed, Follows tend to represent genuine interest. In contrast Facebook engagement is often used to grant approval or because a user desires an association. In addition, the Twitter graph corresponds to actions occurring as far back as 2006 (relatively few edges are ever deleted), while the Facebook graph corresponds only to events since 2014, when we began collecting data. As a result the Twitter data set contains significantly more data, but with less relevance to current events. 

% 3 stages -> figure
\begin{figure}
        \centering
\begin{subfigure}[b]{\hsize}
                \includegraphics[width=\textwidth]{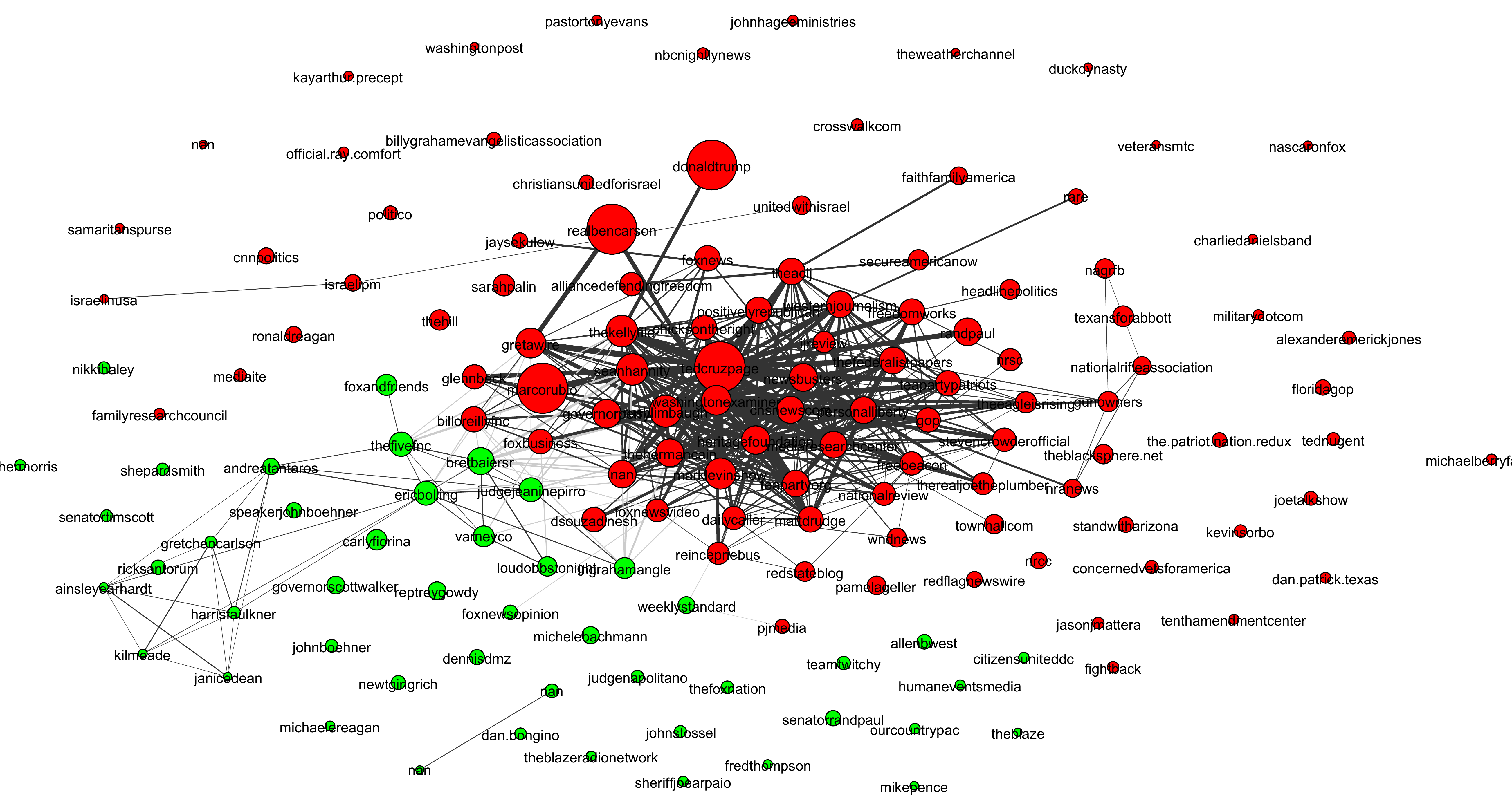}
                \caption{The US republican party. Seeds are ``Donald Trump'', ``Marco Rubio'', ``Ted Cruz'', ``Ben Carson'' and ``Jeb Bush''}
                \label{fig:fb_rep}
        \end{subfigure}
        \begin{subfigure}[b]{\hsize}
                \includegraphics[width=\textwidth,clip]{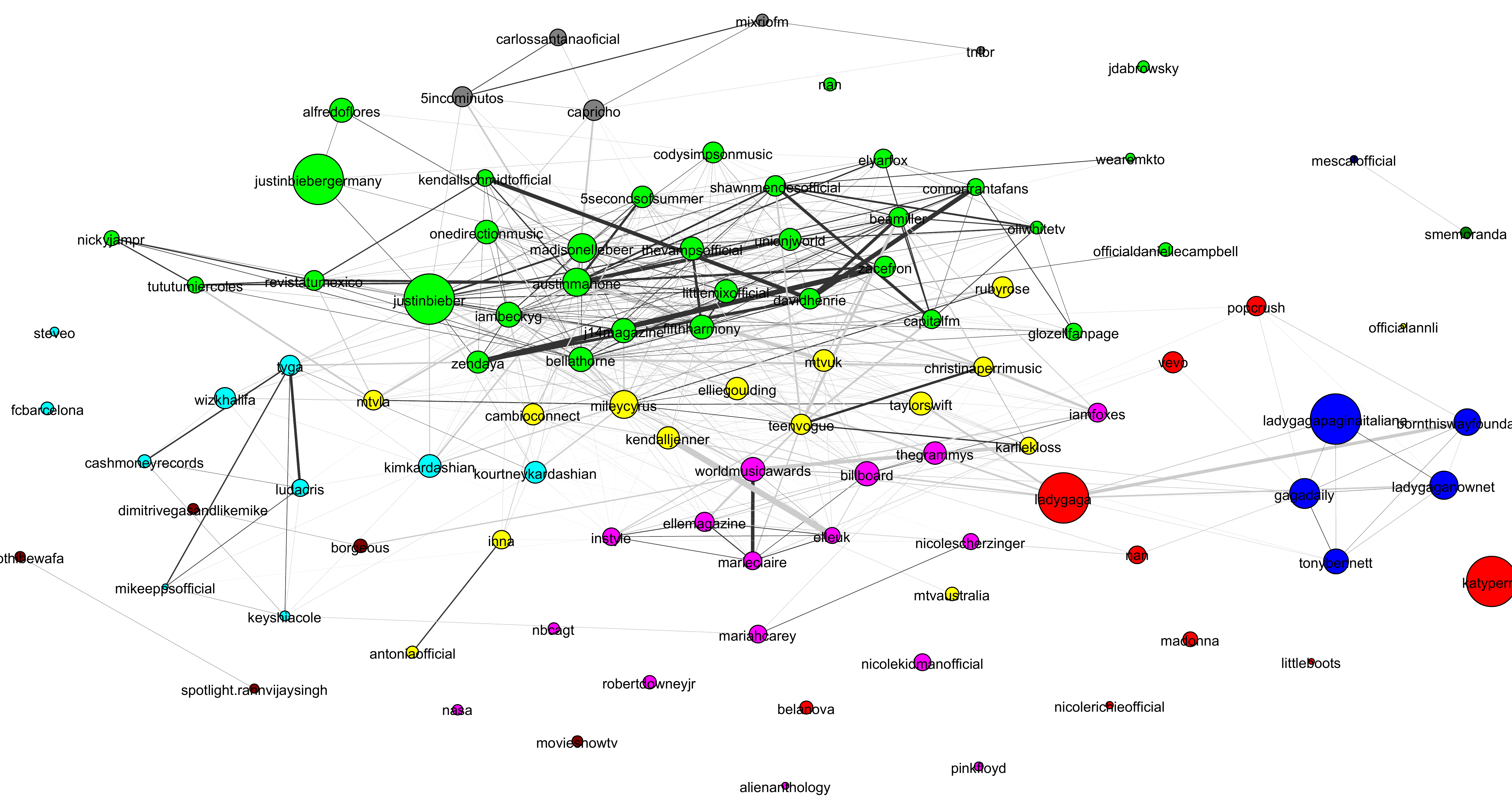}
                \caption{Global Pop Music. Seeds are ``Justin Bieber'', ``Lady Gaga'' and ``Katy Perry''}
                \label{fig:fb_pop}
        \end{subfigure}     

\caption{Visualisations of the \textbf{Facebook Pages engagement graph} around different sets of seed vertices. The vertex size depicts degree of similarity to the seeds. Edge widths show pairwise similarities. Colours are used to show different communities.}
\label{fig:fb_graphs}
%\figspace
\end{figure}
\begin{figure}[tb]
        \centering
        \begin{subfigure}[b]{\textwidth}
                \includegraphics[width=\textwidth,clip]{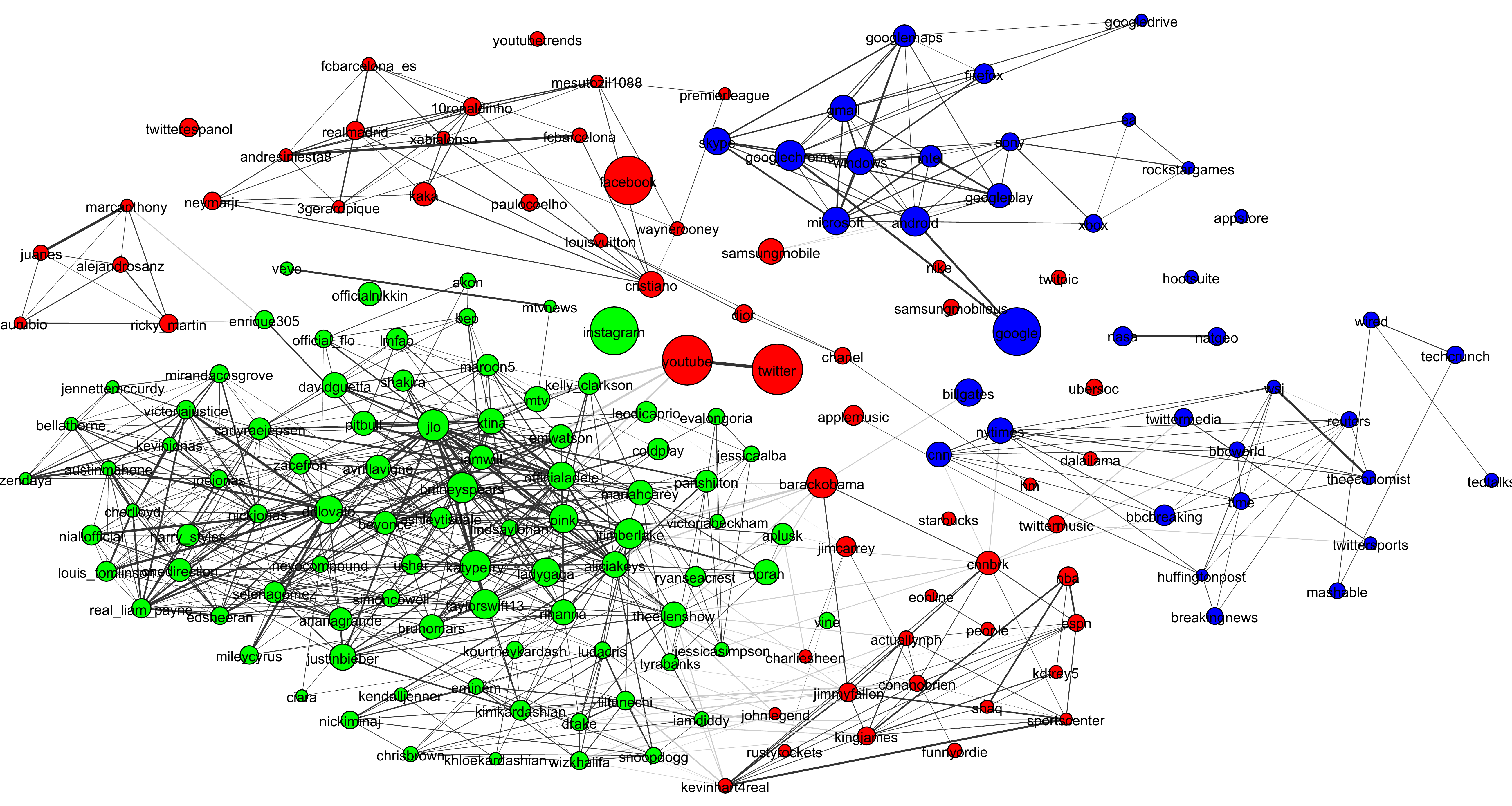}
                \caption{The major social networks. Seeds are Twitter, Facebook, YouTube and Instagram.}
                \label{fig:twit_social}
        \end{subfigure}       
        ~
        \begin{subfigure}[b]{\textwidth}
                \includegraphics[width=\textwidth,clip]{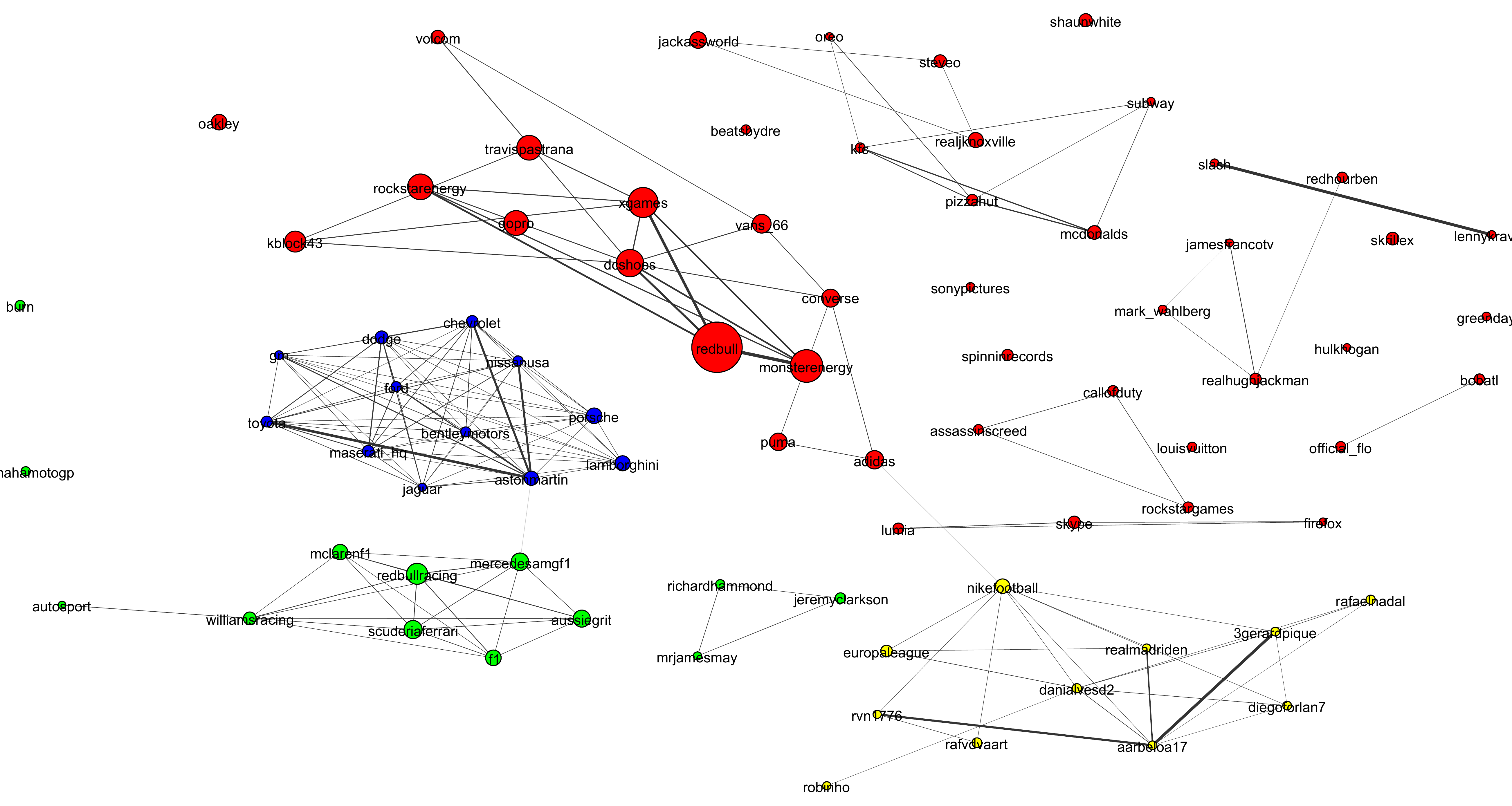}
                \caption{The many faces of RedBull}
                \label{fig:twit_redbull}
        \end{subfigure}               
\caption{Visualisations of the \textbf{Twitter Follower graph} around different sets of seed vertices. Vertex size depicts degree of similarity to the seeds. Edge widths show pairwise similarities. Colours are used to show different communities.}
\label{fig:twit_graphs1}
%\figspace
\end{figure}

\begin{figure}
\begin{subfigure}[tb]{\hsize}
                \includegraphics[width=\textwidth]{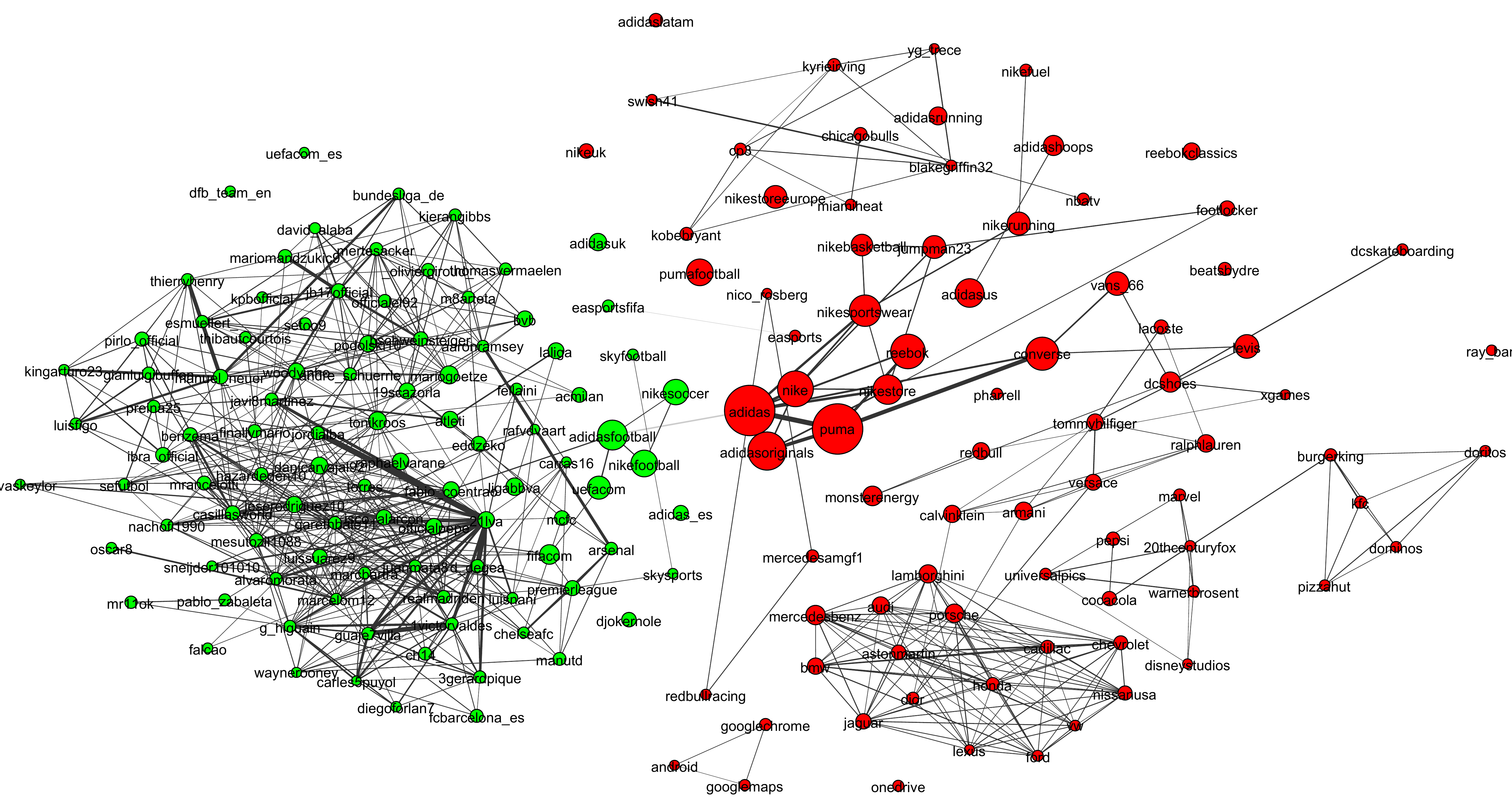}
                \caption{European sport brands. Seeds are Adidas and Puma}
                \label{fig:twit_euro_sport}
        \end{subfigure}
        \begin{subfigure}[b]{\hsize}
                \includegraphics[width=\textwidth,clip]{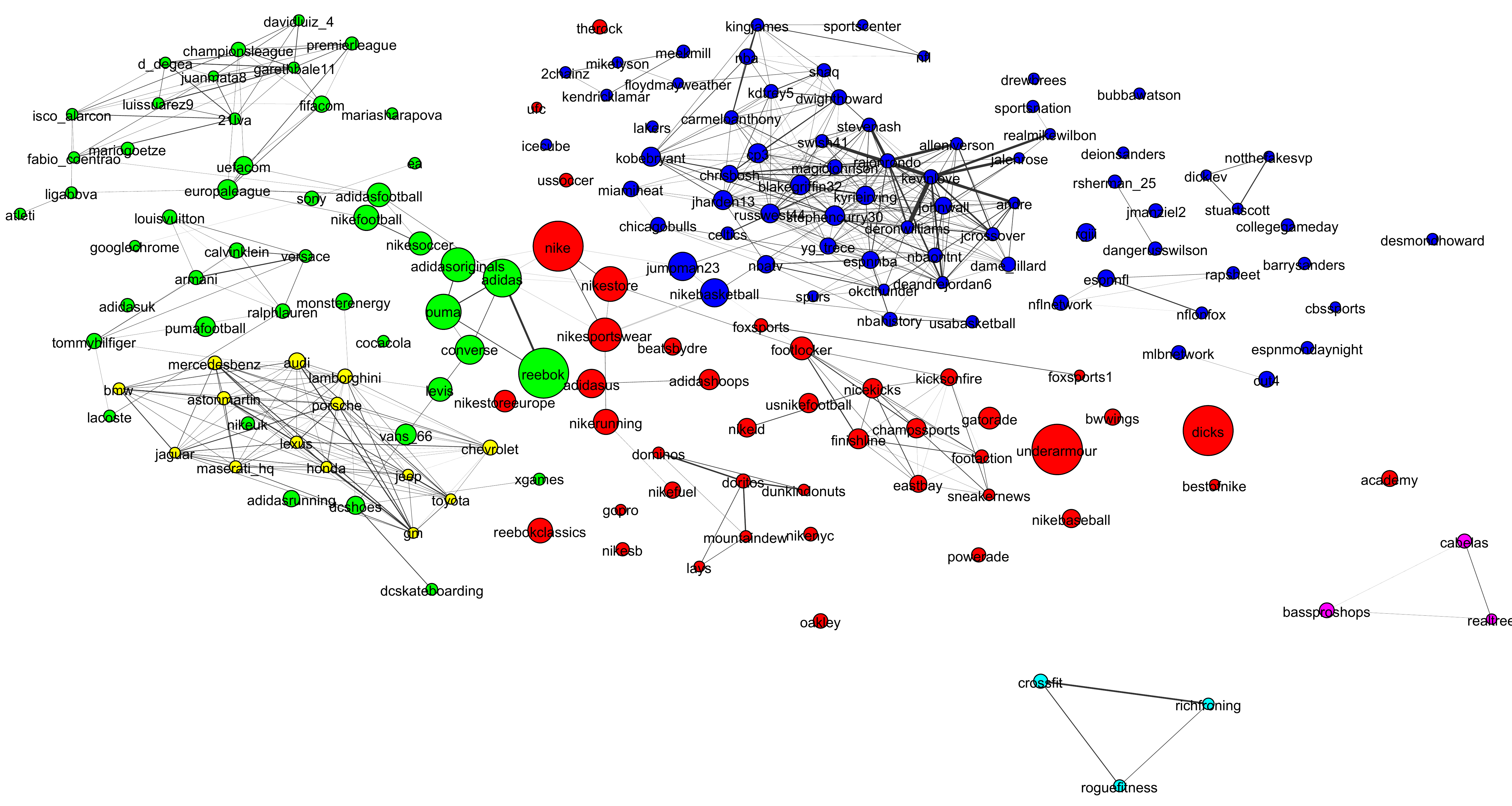}
                \caption{US sports brands. Seeds are Nike, Reebok, UnderArmour, Dicks}
                \label{fig:twit_us_sport}
        \end{subfigure}   
        \caption{Visualisations of the \textbf{Twitter Follower graph} around different sets of seed vertices. Vertex size depicts degree of similarity to the seeds. Edge widths show pairwise similarities. Colours are used to show different communities.}
\label{fig:twit_graphs}
%\figspace
\end{figure}  

% different query types
Our work uses the vast scale and richness of social media data to provide insights into a broad range of questions. Here are some illustrative examples:
\begin{itemize}
% Republican party
\item \textbf{How would you describe the factions and relationships within the US Republican party?}  This is a question with a major temporal component, and so we use the Facebook Pages graph. We feed ``Donald Trump'', ``Marco Rubio'', ``Ted Cruz'', ``Ben Carson'' and ``Jeb Bush'' as seeds into the system and wait for $\unit[0.25]{s}$ for Figure~\ref{fig:fb_rep}, which shows a densely connected core group of active politicians with Donald Trump at the periphery surrounded by a largely disconnected set of right-wing interest bodies.
% pop music
\item \textbf{Which factions exist in global pop music?} We feed the seeds ``Justin Bieber'', ``Lady Gaga'' and ``Katy Perry'' into the system loaded with the Facebook Pages engagement graph and wait for $\unit[0.25]{s}$ for Figure~\ref{fig:fb_pop}, which shows that the industry forms communities that group genders and races.
% social media
\item \textbf{How are the major social networks used?} We feed the seeds ``Twitter'', ``Facebook'', ``YouTube'' and ``Instagram'' into the system loaded with the Twitter Followers graph and wait for $\unit[0.25]{s}$ for Figure~\ref{fig:twit_social}, which shows that Google is highly associated with other technology brands while Instagram is closely related to celebrity and YouTube and Facebook are linked to sports and politics. 
% red bull
\item \textbf{How is the brand RedBull perceived by Twitter users?} We feed the single seed ``RedBull'' into the Twitter Followers graph and wait for $\unit[0.25]{s}$ for Figure~\ref{fig:twit_redbull}, which shows that RedBull has strong associations with motor racing, sports drinks, extreme sports, gaming and football.
% sports brand marketing
\item \textbf{How does sports brand marketing differ between the USA and Europe?} We use the Twitter Followers graph. ``Adidas'' and ``Puma'' are the seeds for the European brands while ``Nike'', ``Reebok'', ``UnderArmour'' and ``Dicks'' are used to represent the US sports brands. Figures~\ref{fig:twit_euro_sport} and~\ref{fig:twit_us_sport} show the enormous importance of football (soccer) to European sports brands, whereas US sports brands are associated with a broad range of sports including hunting, NFL, basketball, baseball and mixed martial arts (MMA).
% \item How is the automotive industry structured? Figures \ref{fig:twit_us_car} and \ref{fig:twit_euro_car} show that the car industry has a core of large global players, that are closely related to specialist media titles \todo[inline]{could drop this}
\end{itemize}

In all cases, the user selects a group of seeds (or a single seed) and runs the system, which returns a Figure and a table of community memberships in $\unit[0.25]{s}$. Analysts can then use the results to supplement the seed list with new entities or use the table of community members from a single WALKTRAP sub-community to explore higher resolution. 

Similar tasks are traditionally conducted with expensive and difficult to scale techniques, such as telephone polling and focus groups, which often take months to return results. In contrast, we are able to produce an automatic analysis in a fraction of a second and at minimal cost, which allows for interactive community detection in large social networks.

\section{Conclusion and Future Work}
\label{sec:conclusion}

% summary statement
We have presented a real-time system to automatically detect communities in large social networks. The system is computationally and memory efficient that it runs on a standard laptop. This work represents a technical advance leading to performance gains that are useful in practice and contains a rigorous evaluation on large social media data sets. The key contributions of this article are  to demonstrate that (1) using the Jaccard similarity of neighbourhood graphs provides a robust association metric between vertices of noisy social networks; (2) Working with minhash signatures of the neighbourhood graph dramatically reduces the space and time requirements of the system with acceptable approximation error; (3) Applying Locality Sensistive Hashing allows for approximate local community detection on very large graphs in real time with acceptable approximation error. For interactive and real-time community detection, we have demonstrated that our system finds higher quality communities in less time than the state-of-the-art algorithm operating under the constraints of a single machine. Our work has clear applications for knowledge discovery processes that currently rely upon slow and expensive manual procedures, such as focus groups and telephone polling. In general, our system offers the potential for organisations to rapidly acquire knowledge of new territories and supplies an alternative monetisation scheme for data owners.

% Other applications
In this article, we focussed on digital social networks, but our method is applicable to all large networks including bipartite networks. The user-item bipartite networks that are studied in the field of recommender systems would be particularly amenable to this treatment, where items could be compactly modelled as minhash signatures of the users who have purchased them.

% Extensions
We leave two extensions for future work. Firstly we treat the input social network as binary. In many settings, information is available to weight the edges. This might include message counts, the time since a connection was made or the type of connection. Efficient methods already exist for working with minhashes of weighted sets~\cite{Manasse2010}. Therefore, an interesting progression of this work is to incorporate data with edges that can contain counts, weights and categorisations.
% Additional compression schemes
The second extension incorporates some of the latest developments in the theory of minhashing. b-bit minhashing and Odd Sketches provide two promising approaches to extend our system to even larger graphs~\cite{Li2010, Mitzenmacher2014}. Both offer the best cost/benefit trade-off when sets are very similar (Jaccard similarity $\approx 1$) or when sets contain most of the elements in the sample space. DSN data typically contains sets that are very small relative to the sample space\footnote{Our Twitter data has a sample space containing $7 \times 10^8$ elements with a typical set containing $10^4$ elements.} and with Jaccard similarities $\ll 0.5$. The strong theoretical bounds of these algorithms do not hold in these DSN-typical settings. Therefore, a cost/benefit analysis similar to Section~\ref{sec:eval_exp1} would be required before implementing either in an extension.

% % On Jaccards
% There are many similarity metrics that can be used to compare the vertices of graph structures. We chose the Jaccard similarity because it can be efficiently encoded in minhash signatures. However it is a symmetric function and symmetric similarities may not always be desirable. When two accounts have very different size neighbourhoods and the smaller is a approximately a subset of the larger, the Jaccard is inconsistent with our intuition that the smaller account is closely related to the larger one, while the converse is not true. Another measure relating two sets is their penetration  
% %
% $$P(A,B) = \frac{A\cap B}{B} \neq P(B,A).$$
% %
% The penetration leads to a directed graph 
% %
% $$P(A,B) = \frac{J(A,B)(|A| + |B|)}{J(A,B)|B|},$$
% %
% and so the penetration can easily be calculated from the Jaccard providing the set sizes are known and the minhash apparatus can still be used.

% \todo[inline]{Could still explain that it is relatively easy to add new accounts and explain how. I also need to address that 675k users isn't that large a social network?}
% \todo[inline]{could also expand a bit on the problem of ground truth for community detection}

\section*{Acknowledgements}

This work was partly funded by a Royal Commission for the Exhibition of 1851 Industrial Fellowship. The authors would like to thank Donal Simmie for his work on optimising the minhash generation procedure.

\FloatBarrier

\appendix

\section{Efficient Minhash Generation}
\label{app:minhash_generation}

A naive Python implementation for generating minhash signatures requires six days to run on a desktop computer with 6 physical
(12 logical) Intel i7 5930k @3.5GHz cores and 64GB of RAM. This is prohibitive for nightly updates and so we highly optimised this part of the code base. The code was ported to the Python-to-C bridge project, Cython, which allowed us to add type information and compiler directives to turn off array bounds checking (there is a large amount of array
dereferencing). We stored
the input matrices in contiguous memory, removing any superfluous code
(logging, most inline error checking). The loops were then rewritten to be vectorised by the
SIMD processor. The fully optimised Cython version of the minhash
implementation runs (in parallel on 6 cores) in approximately one
hour.

\subsection{Crawling Social Networks}
\label{app:crawling}
% network bandwidths
\begin{figure}[tb]
  \centering
    \includegraphics[width=0.9\textwidth,clip]{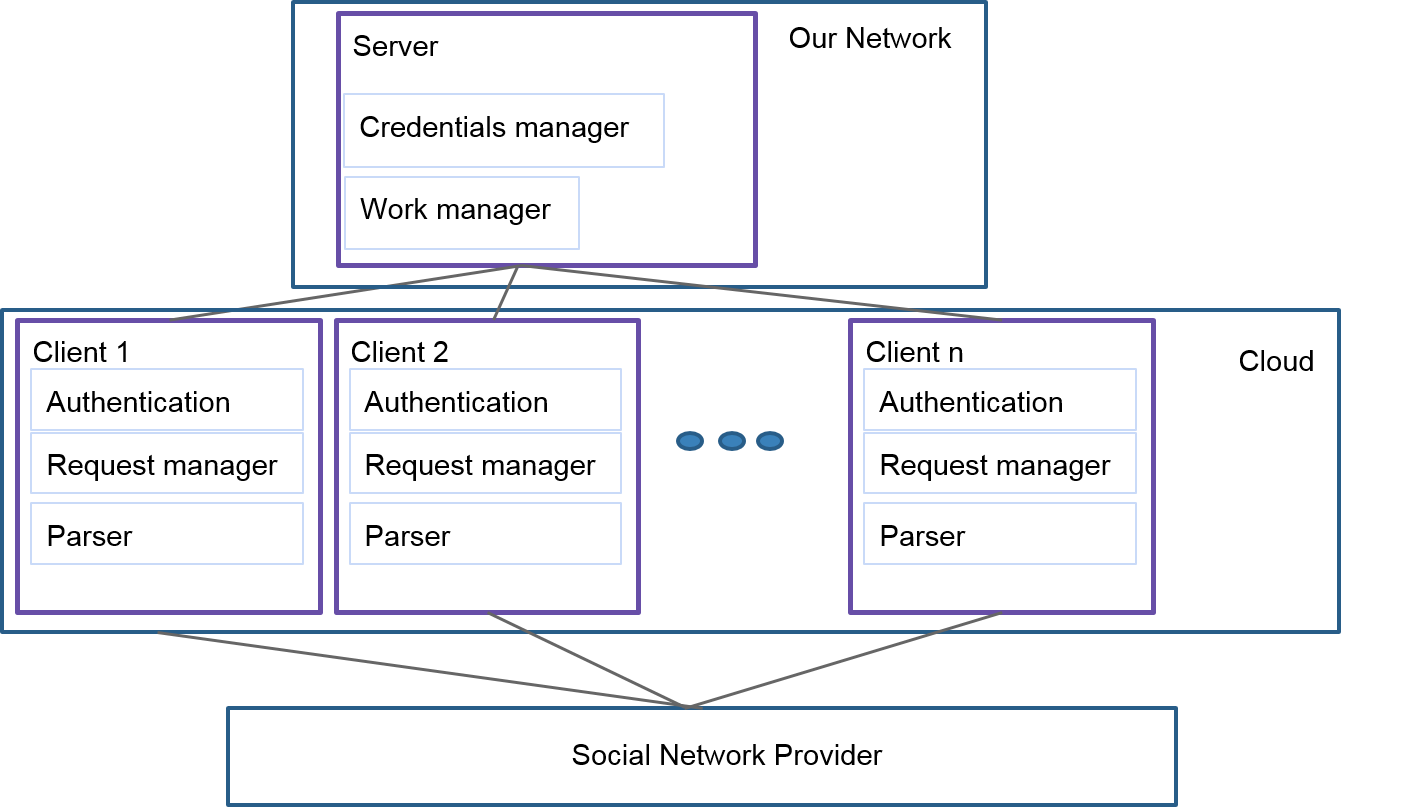}
    \caption{A distributed asynchronous system for data acquisition from digital social networks}
\label{fig:twisted}
\end{figure}

To optimize data throughput while remaining within the DSN rate limits we developed an asynchronous distributed network crawler using Python's Twisted
library~\citep{Wysocki2011}. The crawler consists of a server
responsible for token and work management and a distributed set of clients making
http requests to DSNs (see Figure~\ref{fig:twisted}).

The server contains a credential manager that holds access tokens in a queue and monitors the number of calls to each API. Once a token has been exhausted it is put to the back of the queue and locked until its refresh time. The server communicates over TCP with the clients responding to requests for work and access tokens with account ids and fresh access tokens/pause responses respectively:

% the he need for multithreading
The clients make asynchronous requests to the DSNs, handling response codes, parsing and storing data. A conventional program will \emph{block} while waiting for an http response. When the principal function of a program is to download data, blocking time amounts to the vast majority of the run time. One solution is to run the program using multiple threads. However, for this application threads carry an unnecessary overhead and induce inefficiencies as data is naively moved between caches by the operating system. 
% asynchronous programming
The asynchronous programming paradigm offers a superior alternative to explicit multi-threading. Asynchronous programming makes use of an event loop that constantly \emph{listens} for new jobs and does not block while waiting for http responses. 

% The need to cloud based distributed system
We originally implemented the system using an 80\,MB shared fibre optic connection, but our downloads caused network blackouts. Therefore, we designed a distributed system that could be partially deployed in the cloud. The final system is depicted in Figure~\ref{fig:twisted}. The access tokens and account IDs to query (work) live on a server on our local network. Clients are deployed to Amazon's elastic cloud from where all interactions with DSN servers occurs. We configured the clients to establish persistent connections to the API endpoints. Every time a connection is opened, a handshake must occur. For secure systems (communicating over https), the handshake is particularly onerous, requiring the exchange of security certificates.

\section{Community Axioms}
\label{app:ground_truth}
% The need for community goodness metrics
Homophily only applies to attributes that ease information flow between individuals. Some attributes have no effect or are divisive (for instance right-handed people feel no sense of kinship) and so should not be associated with communities. Additionally attributes may be at the wrong scale to describe structural sub-units (sports person rather than footballer). 
A community evaluation based on ground-truth that were not communities would have no value We apply community goodness functions to each prospective `tag community' to identify to what extent these functional traits generate structurally observable communities. 

% Community goodness metrics explained
For each functional group we generate the fully connected weighted graph by calculating all pairwise Jaccard similarities and evaluate the six metrics in Table~\ref{tab:communities}. They are adapted from \cite{Yang2013} to apply to weighted graphs. As we work with a derived graph where each edge weight is the Jaccard similarity of neighbourhoods, the metrics have slightly different interpretations. Two entities in the derived graph are strongly connected if they have very similar neighbourhoods. Since for the large entities the neighbourhood normally has at least an order of magnitude more incoming than outgoing edges, entities are closely related if they have a similar fan/follower base. We define $S$ to be the set of vertices comprising a community and a weighted graph $G(V,E,W)$ where $W$ is a weight matrix. The internal edge weight of $S$ is
\begin{align}
m_s = \sum_{\{i, j \in S \}}W_{i,j}
\end{align}
and the weight of edges that cross the boundary of $S$ is
\begin{align}
c_s = \sum_{\{i \in S, j \notin S \}}W_{i,j}\,.
\end{align}
The community goodness metrics are then given by:
% list with performance measures
\begin{itemize}
\item \textbf{Clustering} exploits the idea that people in communities are likely to introduce their friends to each other. It measures how cliquey a community is. In our paradigm clustering is high if followers of a community recommend things for other followers of  the community to like or follow. If a vertex has $k_n$ neighbours then $\frac{1}{2}k_n(k_n - 1)$ possible connections can exist between the neighbours. The clustering of a node gives the fraction of its neighbours' possible connections that exist. The clustering of a community is the average clustering of each vertex. Clustering is sometimes referred to as the proportion of triadic closures in the network. The weighted clustering of the $i^{th}$ vertex is given by 
\begin{align}
Cl_i(S) = \frac{W_s^3}{(W_sW_{\max}W_s)_{ii}}
\end{align}
where $W_{\max}$ is a matrix where each entry is the maximum weight found within $S$ \citep{Holme2007}.
\item \textbf{Conductance} is an electrical analogy for how easily information entering the community can leave it. In our context, it is defined as 
\begin{align} \label{eq:con}
Con(S) = \frac{c_s}{2m_s + c_s}\,,
\end{align}
i.e., it is the ratio of the community's external to total edge weight. A low value means that the the community is well separated from the rest of the network. In our paradigm, conductance is low if the followers of the community are not interested in other communities.

\item \textbf{Cohesiveness} measures how easily the community can be split into disconnected components. A good community is not easily broken up. The cohesiveness is given by the minimum conductance of any sub-community. A low value indicates a bad community as there is at least one well-separated sub-community. In our paradigm, low cohesiveness corresponds to members of the community having distinct, non-overlapping follower groups.
\begin{align}
Coh(S) = \min_{\{S' \subset S\}}Con(S')
\end{align}
Iterating through all subsets $S'$ of $S$ is impractical. Thus, we sample $S'$ by randomly selecting 10 subsets of starting vertices, running PPR community detection for each and taking a sweep through the PageRank vector to find the minimum conductance cut.
\item \textbf{Conductance Ratio (CR)} is the ratio of conductance to cohesiveness and defined as
\begin{align}
CR(S) = \frac{Con(S)}{Coh(S)}\,.
\end{align}
A large number indicates that the community could be broken up into structural sub-units. 
\item \textbf{Density} is given by the ratio of the community's total internal edge weight to the maximum possible if every edge was present with weight one: 
\begin{align}
Den = \frac{2m_s}{n_s(n_s - 1)}
\end{align}
A high number indicates a highly interconnected community. In our paradigm, this corresponds to a community with a well-defined follower base that is interested in most community members.
\item \textbf{Separability} measures how well the community is separated from the rest of the network. It is the ratio of internal to external edges and so is closely related to conductance:
\begin{align} \label{eq:sep}
Sep(S) = \frac{m_s}{c_s}
\end{align}
In our paradigm, a high value indicates that followers of the community are not interested in much else.
\end{itemize}

% \section{Penetration}
% \label{app:penetration}

% We derive the relationship between the Jaccard similarity and the penetration.

% $$J(A,B) = \frac{|A \cap B|}{|A \cup B|}$$
% $$P(A,B) = \frac{|A \cap B|}{|B|}$$
% dropping the function arguments for brevity
% $$J|A \cup B| = P|B|$$
% decomposing the union gives
% $$J\frac{|A| + |B| - |A \cap B|}{|B|} = P$$
% using the definition of the penetration
% $$J\frac{|A| + |B|}{|B|} - JP = P$$
% re-arranging gives
% $$P = \frac{J(|A| + |B|)}{|B|(1+J)}$$

\vskip 0.2in
% \bibliography{Mendeley.bib}

\end{document}